\documentclass[12pt]{article}
\usepackage{amsmath}
\numberwithin{equation}{section}
\usepackage{amssymb}
\usepackage{graphicx,color,times,amsthm,extarrows}
\usepackage[all]{xy}

\newtheorem{theorem}{Theorem}
\newtheorem{lemma}{Lemma}

\newtheorem{remark}[theorem]{Remark}
\newtheorem{definition}[theorem]{Definition}
\newtheorem{guess}{Hypothesis}
\makeatletter
\def\equfill@{\arrowfill@\Relbar\Relbar\Relbar}
\newcommand{\equfill}[2][]{\ext@arrow 0395\equfill@{#1}{#2}}
\makeatother
\makeatletter
\def\Eqlfill@{\arrowfill@\Relbar\Relbar\Relbar}
\newcommand{\extendEql}[1][]{\ext@arrow 0359\Eqlfill@{#1}}
\makeatother
\begin{document}
\title{The Algebro-Geometric Initial Value Problem for the  Relativistic Lotka-Volterra
Hierarchy  and  Quasi-Periodic Solutions}
\author{Peng Zhao, Engui Fan\footnote{Corresponding author and E-mail: faneg@fudan.edu.cn},\quad Yu Hou
\vspace{3mm}\\
Institue of Mathematics, Fudan University, Handan Road 220, \\
  {Shanghai 200433, P R China}}
  \date{}
\maketitle
\leftline{\bf\ Abstract}
 We provide a detailed treatment of  relativistic Lotka-Volterra hierarchy and a kind of initial value problem with special emphasis on its the theta function representation of all algebro-geometric solutions. The basic tools involve
 hyperelliptic curve $\mathcal{K}_n$ associated with the Burchnall-Chaundy polynomial, Dubrovin-type
 equations for auxiliary divisors and associated trace formulas. With the help of a foundamental meromorphic
 function $\tilde{\phi}$ on $\mathcal{K}_p$ and trace formulas, the complex-valued algebro-geometric solutions of
 of RLV hierarchy are derived.\\

\section{Introduction}
\qquad Nonlinear integrable lattice systems have been studied extensively in relation with various aspects and they usually possess rich mathematical structure such as Lax pairs, Hamilton structure, conservation law, etc.
The Toda lattice (TL),
\begin{equation}\label{1.1}
\text{TL}:
\begin{cases}
  a_{t}=a(b^+-b),\cr
  b_t=a-a^-,\cr
\end{cases}
\end{equation}
is one of the most important integrable systems. It is well-known soliton equations such as the KdV, modified KdV, and nonlinear Schr\"{o}dinger equations are closely related to or derived from the Toda equation by
suitable limiting procedures \cite{119,120}. Another celebrated integrable lattice system is Lotka-Volterra (LV) lattice,
\begin{equation}\label{1.2}
\text{LV}:
\begin{cases}
  u_t=u(v-v^-),\cr
  v_t=v(u^+-u),\cr
\end{cases}
\end{equation}
or
\begin{equation}\label{1.3}
  a_t=a(a^+-a^-),
\end{equation}
by setting $a(2n-1)=u(n), a(2n)=v(n).$
Ruijsenaars found a relativistic integrable generalization of non-relativistic Toda lattice through  solving a relativistic
version of the Calogero-Moser system \cite{RT}.The Lax representation,
inverse scattering problem of the Ruijsenaars-Toda lattice and its connection with soliton
dynamics were investigated. A general approach to constructing relativistic generalizations
of integrable lattice systems, applicable to the whole lattice KP hierarchy,
was proposed by Gibbons and Kupershmidt \cite{i}.
However, nobody have known what the relativistic Lotka-Volterra  lattice is until
Y. B. Suris and O. Ragnisco found it \cite{RLV}.

There are very close connnection among the Toda lattice, the Lotka-Volterra lattice, the relativistic Toda lattice and the relativistic Lotka-Volterra lattice. It is well known that (\ref{1.1}) and (\ref{1.2}) are both discrete version of the KdV equation in the sense of different limiting process and
in the non-relativistic limit $h\rightarrow 0$, the relativistic Toda system (RT),
\begin{equation}\label{1.3}
\text{RT}:
\begin{cases}
  a_t=a(b^+-b+ha^+-ha^-),\cr
  b_t=(1+hb)(a-a^-),\cr
\end{cases}
\end{equation}
and the  "relativistic splitting" relativistic Lotka-Volterra systems (RLV),
\begin{equation}\label{1.4}
\text{RLV}:
\begin{cases}
   u_t=u(v-v^-+huv-hu^-v^-),\cr
   v_t=v(u^+-u+hu^+v^+-huv),\cr
\end{cases}
\end{equation}
reduced to the well-known Toda lattice equation (\ref{1.1}) and Lotka-Volterra lattice equation (\ref{1.2}), respectively. Moreover,
 the Miura relation can be summarized in the following diagram:
$$\xymatrix{\text{RLV} \ar[r]^{h\rightarrow 0} \ar[d]_{\tau_{1}} & \text{LV} \ar[d]^{\tau_{2}} \\ \text{RT} \ar[r]_{h\rightarrow 0} & \text{TL}}$$
$$\text{RLV}_1\xrightarrow{\tau_1}\text{RT}, \quad\text{RLV}_2\xrightarrow{\tau_2}\text{RT} $$
Here
\begin{equation*}
  \tau_1:
  \begin{cases}
  a=uv,\cr b=u+v^-,\cr
  \end{cases}
\end{equation*}
\begin{equation*}
  \tau_2:
  \begin{cases}
  a=u^+v,\cr b=u+v.\cr
  \end{cases}
\end{equation*}

Mathematical structure related to relativistic volterra lattice (\ref{1.4}) 
such as Lax integrablity \cite{300}, $2\times 2$ Lax representation \cite{Laxpresentation}, conservation laws \cite{zwxz}, bilinear structure and determinant solution \cite{KMMO} have been closely studied and hence the purpose of this paper is to uniformly construct algebro-geometric
solutions of the relativistic Lotka-Volterra hierarchy which invariably is connected with geometry and Riemann theta functions, parameterized by some Riemann surface. Algebro-geometric solutions (finite-gap solutions or
quasi-period solutions),  as  an important character of integrable system,
 is  a kind of explicit solutions closely related to the inverse spectral theory \cite{99,t}.
Around 1975, several independent groups in UUSR and USA, namely, Novikov, Dubrovin and Krichever in Moscow, Matveev and Its in Leningrad, Lax, McKean, van Moerbeke and
M. Kac in New York, and Marchenko, Kotlyarov and Kozel in Kharkov, developed the so-called finite finite-gap theory of
nonlinear  KdV   equation based on
the works of Drach, Burchnall and Chaunchy, and Baker \cite{21,5}.  The algebro-geometric method they established allowed us to find an important
class of exact solutions to the soliton equations.
As a degenerated case of this solutions, the multisoliton solutions and elliptic functions may be obtained \cite{t,1}.
Its and Matveev first derived explicit expression of the quasi-period solution of KdV equation in 1975 \cite{4}, which is closely related to
the finite-gap spectrum of the associated differential operator.  Further exciting results appeared later, including the finite-gap solutions of
  Toda lattice,  the Kadomtsev-Petviashvili  equation and others \cite{d,5,1}, which could be
found in the wonderful work of Belokolos, et al \cite{t}.  In recent years, a systematic approach based on the nonlinearization technique
of Lax pairs or the restricted flow technique to derive the algebro-geometric solutions of (1+1)-  and (2+1)-dimensional soliton equations has been obtained
\cite{7,11}. An alternate systematic approach proposed by Gesztesy and Holden can be used to construct  algebro-geometric solutions has been extended to the whole (1+1) dimensional continuous and discrete hierarchy  models  \cite{A1,A2,u,19,13}.

The outline of our present paper is as follow. In section 2,
the relativistic Lotka-Volterra equation (\ref{1.4}) is extended to a whole hierarchy through the polynomial recursive relation. The hyperelliptic curve
associated with RLV hierarchy is given in terms of the polynomial. In section 3, We focus on the stationary RLV hierarchy. based on the polynomial
recursion formalism introduced in Section 2 and a fundamental meromorphic function $\tilde{\phi}$ on the hyperelliptic curve $\mathcal{K}_n$, we study the Baker-Akhiezer function $\Psi,$ trace formulas, from which the algebro-geometric solutions for stationary RLV hierarchy are constructed in terms of
Riemann theta functions. In section 4, we extend the algebro-geometric
analysis of Section 3 to the time-dependent RLV hierarchy based on a kind of special initial value problem.
Finally, in section 5  we give  Lagrange interpolation representation that will be  used in this paper.

\section{ Relativistic Lotka-Volterra hierarchy and associated hyperelliptic curve}
\qquad In this section, we investigate the relativistic Lotka-Volterra hierarchy and then derive
the hyperelliptic algebraic curves associated with the algebro-geometric solutions of the newly
constructed hierarchy. Throughout this paper, we have the following definition.
\begin{definition} We denote by $\ell$($\mathbb{Z}$) the set of all the complex-valued sequences
 $\{f(n)\}_{n=-\infty}^{+\infty}$. This is a vector space with respect to the naturally defined operation.
 A subspace $\ell^{2}(\mathbb{Z})\subset\ell(\mathbb{Z})$ is defined by the set of
\{$f\in\ell(\mathbb{Z})|\sum^{+\infty}_{n=-\infty}|f(n)|^{2}<+\infty,n\in\mathbb{Z}$\}.
\end{definition}
\begin{definition} We denote by $S^{\pm}$ the shift operators acting on
$\psi=\{\psi(n)\}_{n=-\infty}^{+\infty}\in\ell(\mathbb{Z})$ according to $(S^{\pm}\psi)(n)=\psi(n\pm1)$.
 The identity operator I acting on $\psi=\{\psi(n)\}_{n=-\infty}^{+\infty}\in\ell(\mathbb{Z})$ according
to $(I\psi)(n)=\psi(n)$. We also define $\psi^{\pm}=S^{\pm}\psi,\psi\in\ell(\mathbb{Z})$.
\end{definition}

 We introduce the following $2\times2$ matrix problem \cite{Laxpresentation}
 \begin{equation}\label{2.1}
 \begin{split}
   &S^{+}\Psi=U(\lambda)\Psi,\\
   &U(\lambda)=\left(
      \begin{array}{cc}
       \lambda \tilde{p}-\lambda^{-1} & \tilde{q}  \\
        \tilde{r} & \lambda \\
       \end{array}
    \right), \quad \lambda\in\mathbb{C},
 \end{split}
 \end{equation}
   where $\Psi=(\Psi_1,\Psi_2)^{T}\in\ell(\mathbb{Z})\times\ell(\mathbb{Z}), \tilde{r}=(\tilde{p}-1)\tilde{q}^{-1}, \tilde{p}=p, \tilde{q}=e^{(S^+-I)^{-1}\ln q}$ and $\tilde{p}, \tilde{q}, p, q$ are potential functions and $\lambda$ is the spectral parameter.
   Here $p=p(n,t),q=q(n,t)\in\ell(\mathbb{Z}), (n,t)\in\mathbb{Z}\times\mathbb{R}$ in time-dependent case and $p=p(n),q=q(n)\in\mathit{l}(\mathbb{Z}), n\in\mathbb{Z}$ in stationary case.

   Define sequences $\{a_\ell(n)\}_{\ell\in\mathbb{N}_0}, \{b_\ell(n)\}_{\ell\in\mathbb{N}_0}$ and $\{c_\ell(n)\}_{\ell\in\mathbb{N}_0}$ in $\ell(\mathbb{Z})$ recursively by
   \begin{equation}\label{2.2}
    \tilde{p}a_\ell^{-}-\tilde{p}a_\ell=a_{\ell+1}^--a_{\ell+1}+\tilde{r}b_{\ell}-\tilde{q}c_{\ell}^{-},\quad\ell\in\mathbb{N}_{0},
   \end{equation}
   \begin{equation}\label{2.3}
    \tilde{p}b_\ell^{-}-b_{\ell+1}^{-}-\tilde{q}a_{\ell+1}^{-}=\tilde{q}a_{\ell+1}+b_{\ell},\quad\ell\in\mathbb{N}_{0},
   \end{equation}
   \begin{equation}\label{2.4}
    \tilde{r}a_{\ell+1}^{-}+c_{\ell}^{-}=\tilde{p}c_{\ell}-c_{\ell+1}-\tilde{r}a_{\ell+1},\quad\ell\in\mathbb{N}_{0},
   \end{equation}
   \begin{equation}\label{2.5}
    a_0=1/2, b_{0}=-\tilde{q}^{+}, c_0=-\tilde{r}.
   \end{equation}
    Explicitly, one obtains
   \begin{equation}\label{2.6}
     \begin{split}
       &a_1=-\tilde{q}^{+}\tilde{r}+\delta_1/2,\\
       &b_1=-\tilde{p}^{+}\tilde{q}^{+}+(\tilde{q}^{+})^{2}\tilde{r}
       +\tilde{q}^{+}\tilde{q}^{++}\tilde{r}^{+}+\tilde{q}^{++}-\tilde{q}^{+}\delta_1,\\
       &c_1=-\tilde{p}\tilde{r}+\tilde{q}^{+}\tilde{r}^{2}+\tilde{r}^{-}+\tilde{q}\tilde{r}\tilde{r}^{-}-\tilde{r}\delta_1,\\
       &a_2=-\tilde{p}^+\tilde{q}^+\tilde{r}+(\tilde{q}^+)^2\tilde{r}^2-\tilde{p}\tilde{q}^+\tilde{r}+\tilde{q}^{++}\tilde{r}
       +\tilde{q}^+\tilde{r}^{-}+\tilde{q}^{++}\tilde{q}^+\tilde{r}^+\tilde{r}+\tilde{q}^+\tilde{q}\tilde{r}\tilde{r}^-\\
       &-\tilde{q}^+\tilde{r}\delta_1+\delta_2/2,\\
       &\dotsi \dotsi.
     \end{split}
   \end{equation}
   Here $\{\delta_{\ell}\}_{\ell\in\mathbb{N}}$ denote summation constants which naturally arise when solving (\ref{2.2})-(\ref{2.5}).

   If we denote by $\bar{a}_{\ell}=a_{\ell}|_{\delta_{j}=0,j=1,\dotsi,\ell}, \bar{b}_{\ell}=b_{\ell}|_{\delta_{j}=0,j=1,\dotsi,\ell}, \bar{c}_{\ell}=c_{\ell}|_{\delta_{j}=0,j=1,\dotsi,\ell}$, $\ell\in\mathbb{N}$ the homogeneous coefficients of $a_\ell, b_\ell, c_\ell$, that is,
   \begin{equation}\label{2.7}
     \begin{split}
       &\bar{a}_1=-\tilde{q}^{+}\tilde{r},\\
       &\bar{b}_1=-\tilde{p}^{+}\tilde{q}^{+}+(\tilde{q}^{+})^{2}\tilde{r}
       +\tilde{q}^{+}\tilde{q}^{++}\tilde{r}^{+}+\tilde{q}^{++},\\
       &\bar{c}_1=-\tilde{p}\tilde{r}+\tilde{q}^{+}\tilde{r}^{2}+\tilde{r}^{-}+\tilde{q}\tilde{r}\tilde{r}^{-},\\
       &\bar{a}_2=-\tilde{p}^+\tilde{q}^+\tilde{r}+(\tilde{q}^+)^2\tilde{r}^2-\tilde{p}\tilde{q}^+\tilde{r}+\tilde{q}^{++}\tilde{r}
       +\tilde{q}^+\tilde{r}^{-}+\tilde{q}^{++}\tilde{q}^+\tilde{r}^+\tilde{r}+\tilde{q}^+\tilde{q}\tilde{r}\tilde{r}^-\\
       &\dotsi\dotsi.
     \end{split}
   \end{equation}
   By induction one infers that
   \begin{equation}\label{2.8}
       a_{\ell}=\sum_{k=0}^{\ell}\delta_{\ell-k}\bar{a}_{k},\quad \ell\in\mathbb{N}_0,
   \end{equation}
   \begin{equation}\label{2.9}
       b_{\ell}=\sum_{k=0}^{\ell}\delta_{\ell-k}\bar{b}_{k},\quad \ell\in\mathbb{N}_0,
   \end{equation}
   \begin{equation}\label{2.10}
       c_{\ell}=\sum_{k=0}^{\ell}\delta_{\ell-k}\bar{c}_{k},\quad \ell\in\mathbb{N}_0,
   \end{equation}
   introducing $\delta_0=1$.

   To construct the relativistic Lotka-Volterra hierarchy we will consider the following
   ansatz
   \begin{equation}\label{2.11}
   V_n(\lambda)=\left(
        \begin{array}{cc}
          A_{2n+2}^-(\lambda) & B_{2n+1}^-(\lambda) \\
          C_{2n+1}^-(\lambda) & -D_{2n+2}^-(\lambda)
        \end{array}\right),\quad n\in\mathbb{N}_0,
   \end{equation}
   where $A_{2n+2}, B_{2n+1}, C_{2n+1}$ and $D_{2n+2}$ are chosen as polynomials, namely
   \begin{equation}\label{2.12}
       A_{2n+2}(\lambda)=\sum_{\ell=0}^{n+1}a_\ell\lambda^{-\left(2n+2-2\ell\right)}=
       a_0\lambda^{-(2n+2)}+a_1\lambda^{-2n}+\dotsi+a_{n+1},
   \end{equation}
   \begin{equation}\label{2.13}
       B_{2n+1}(\lambda)=\sum_{\ell=0}^{n}b_\ell\lambda^{-\left(2n+1-2\ell\right)}=
       b_0\lambda^{-(2n+1)}+b_1\lambda^{-(2n-1)}+\dotsi+b_{n}\lambda^{-1},
   \end{equation}
   \begin{equation}\label{2.14}
       C_{2n+1}(\lambda)=\sum_{\ell=0}^{n}c_\ell\lambda^{-\left(2n+1-2\ell\right)}=
       c_0\lambda^{-(2n+1)}+c_1\lambda^{-(2n-1)}+\dotsi+c_{n}\lambda^{-1},
   \end{equation}
   \begin{equation}\label{2.15}
      D_{2n+2}(\lambda)=\sum_{\ell=0}^{n+1}d_\ell\lambda^{-\left(2n+2-2\ell\right)}=
       d_0\lambda^{-(2n+2)}+d_1\lambda^{-2n}+\dotsi+d_{n+1}
   \end{equation}
   and the coefficient $\{a_\ell\}_{\ell=0}^{n+1}, \{b_\ell\}_{\ell=0}^{n}, \{c_\ell\}_{\ell=0}^{n}, \{d_\ell\}_{\ell=0}^{n+1}$ are defined in $\ell(\mathbb{Z})$.

   The stationary zero-curvature equation
   \begin{equation}\label{2.16}
      0=U(\lambda)V_n(\lambda)-V_n^+(\lambda)U(\lambda)
   \end{equation}
   is equivalent to the following equalities
   \begin{equation}\label{2.17}
     (\tilde{p}\lambda-\lambda^{-1})A_{2n+2}^-+\tilde{q}C^{-}_{2n+1}=
     (\tilde{p}\lambda-\lambda^{-1})A_{2n+2}+\tilde{r}B_{2n+1},
   \end{equation}
   \begin{equation}\label{2.18}
     (\tilde{p}\lambda-\lambda^{-1})B_{2n+1}^--\tilde{q}D_{2n+2}^-=\tilde{q}A_{2n+2}+B_{2n+1}\lambda,
   \end{equation}
   \begin{equation}\label{2.19}
     \tilde{r}A_{2n+2}^-+C_{2n+1}^-\lambda=(\tilde{p}\lambda-\lambda^{-1})C_{2n+1}-\tilde{r}D_{2n+2},
   \end{equation}
   \begin{equation}\label{2.20}
     \tilde{r}B_{2n+1}^--D_{2n+2}^-\lambda=\tilde{q}C_{2n+1}-D_{2n+2}\lambda.
   \end{equation}
   From (\ref{2.16}), one finds the matrix $V_n(\lambda)$ is similar to $V_{n}^{+}(\lambda)$ and then we have
   $$\text{trace}\left(V_n(\lambda)\right)=\text{trace}\left(V_{n}^+(\lambda)\right),$$
   namely, $$A^{-}_{2n+2}-D^{-}_{2n+2}=A_{2n+2}-D_{2n+2}.$$
   Hence $A_{2n+2}-D_{2n+2}$ is $n$-independence.
   Without loss of generality we can choose
   \begin{equation}\label{2.21}
       d_{\ell}=a_{\ell},\quad \ell=0,1,2,\dotsi,n, \quad d_{n+1}=\triangle,
   \end{equation}
   where $\triangle=1=(\dotsi,1,1,1,\dotsi)\in\ell(\mathbb{Z})$ is a constant value sequence.
   In stationary case we have \begin{equation}\label{2021}A_{2n+2}=D_{2n+2},\end{equation}
   since this can be always be achieved by adding a constant coefficient polynomial times the identity matrix to
   $V_n$, which not affect the stationary zero-curvature equation (\ref{2.16}).

   Plugging the ansatz (\ref{2.12})-(\ref{2.15}) (\ref{2.21}) into (\ref{2.17})-(\ref{2.20}) and comparing the coefficients yields
   the following relations for $\{a_{\ell}\}_{\ell=0}^{n+1},\{b_{\ell}\}_{\ell=0}^{n},\{c_{\ell}\}_{\ell=0}^{n}.$
   That is,
   \begin{equation}\label{2.22}
        \tilde{p}a_\ell^{-}-\tilde{p}a_\ell=a_{\ell+1}^--a_{\ell+1}+\tilde{r}b_{\ell}-\tilde{q}c_{\ell}^{-},\quad\ell=0,1,\dotsi,n-1,
   \end{equation}
   \begin{equation}\label{2.23}
        \tilde{p}b_\ell^{-}-b_{\ell+1}^{-}-\tilde{q}a_{\ell+1}^{-}=\tilde{q}a_{\ell+1}+b_{\ell},\quad\ell=0,1,\dotsi,n-1,
   \end{equation}
   \begin{equation}\label{2.24}
       \tilde{r}a_{\ell+1}^{-}+c_{\ell}^{-}=\tilde{p}c_{\ell}-c_{\ell+1}-\tilde{r}a_{\ell+1},\quad\ell=0,1,\dotsi,n-1,
   \end{equation}
   \begin{equation}\label{2.25}
       \tilde{r}b_{\ell+1}^--a_{\ell+1}^-=\tilde{q}c_{\ell+1}-a_{\ell+1},\quad\ell=0,\dotsi,n-1,
   \end{equation}
   \begin{equation}\label{2.26}
       \tilde{p}a_{n+1}^--\tilde{p}a_{n+1}=0,
   \end{equation}
   \begin{equation}\label{2.27}
       \tilde{p}b_{n}^--\tilde{q}\triangle=\tilde{q}a_{n+1}+b_{n},
   \end{equation}
   \begin{equation}\label{2.28}
       \tilde{r}a_{n+1}^-+c_n^-=\tilde{p}c_n-\tilde{r}\triangle,
   \end{equation}
   \begin{equation}\label{2.29}
       -\triangle^-=-\triangle.
   \end{equation}
   Hence if $\{a_{\ell}\}_{\ell=0}^{n+1},\{b_{\ell}\}_{\ell=0}^{n},\{c_{\ell}\}_{\ell=0}^{n}$ are defined by
   (\ref{2.2})-(\ref{2.5}), then one finds (\ref{2.22})-(\ref{2.24}) (\ref{2.29}) hold naturally.
   From (\ref{2.22})-(\ref{2.24}), one calculate
   \begin{equation}\label{2.30}
      \begin{split}
         &\tilde{r}b_{\ell+1}-\tilde{q}a_{\ell+1}\\
         &=\tilde{r}\left(\tilde{p}b_{\ell}^--\tilde{q}a_{\ell+1}^--\tilde{q}a_{\ell+1}-b_{\ell}\right)-\tilde{q}
         \left(\tilde{p}c_{\ell}-\tilde{r}a_{\ell+1}-\tilde{r}a_{\ell+1}-c_{\ell}^-\right)\\
         &=\tilde{p}(\tilde{r}b_{\ell}^--\tilde{q}c_{\ell})-\tilde{r}b_{\ell}+\tilde{q}c_{\ell}^-\\
         &=\tilde{p}(a_{\ell}^--a_{\ell})-\tilde{r}b_{\ell}+\tilde{q}c_{\ell}^-\\
         &=a_{\ell+1}^--a_{\ell+1},\\
      \end{split}
   \end{equation}
   where we used the induction
   $$\tilde{r}b_{k}^--a_{k}^-=\tilde{q}c_{k}-a_{k},\quad k=0,1,\dotsi,\ell$$
   in the third equality of (\ref{2.30}).
   Therefore (\ref{2.25}) is the direct result of (\ref{2.22})-(\ref{2.24}).
   Using (\ref{2.2})-(\ref{2.5}), (\ref{2.26})-(\ref{2.29}) are equivalent to
   \begin{equation}\label{2.31}
        \tilde{p}a_{n+1}^--\tilde{p}a_{n+1}=0,
   \end{equation}
   \begin{equation}\label{2.32}
        b_{n+1}^{-}+\tilde{q}a_{n+1}^--\tilde{q}\triangle=0,
   \end{equation}
   \begin{equation}\label{2.33}
       c_{n+1}+\tilde{r}a_{n+1}+\tilde{r}\triangle=0.
   \end{equation}
   Noticing the condition (\ref{2.25}) and above analysis, (\ref{2.31}) and (\ref{2.32})
   give rise to the stationary relativistic Lotka-Volterra hierarchy, which we introduce
   as follows
   \begin{equation}\label{2.34}
      \text{s-RLV}_n(p,q)=\text{s-}\widetilde{\text{RLV}_n}(\tilde{p},\tilde{q},\tilde{r})
      =\left(
         \begin{array}{c}
           a_{n+1}^--a_{n+1} \\
           b_{n+1}^{-}+\tilde{q}a_{n+1}^--\tilde{q}\triangle \\
         \end{array}
       \right)=0,\quad p\in\mathbb{N}_0.
   \end{equation}
   Explicitly,
   \begin{equation}\label{2.35}
       \begin{split}
        &\text{s-RLV}_0(p,q)=\text{s-}\widetilde{\text{RLV}_0}(\tilde{p},\tilde{q},\tilde{r})\\
        &=\left(\begin{array}{c}
           -\tilde{q}\tilde{r}^-+\tilde{q}^+\tilde{r} \\
           -\tilde{p}\tilde{q}+\tilde{q}\tilde{q}^+\tilde{r}+\tilde{q}^+-\tilde{q}(\triangle+\delta_1/2)
         \end{array}\right)\\
         &=\left(
             \begin{array}{c}
               pq-p^-q^-+q^--q \\
               \tilde{q}\left(pq-p-(\delta_1/2+\triangle)\right) \\
             \end{array}
           \right)=0,
         \end{split}
   \end{equation}
   where we use the relation
   \begin{equation}\label{2.36}
       \begin{split}
       &\tilde{q}^+/\tilde{q}=e^{(S^+-I)^{-1}q^+}/e^{(S^+-I)^{-1}q}=q,\\
       &\tilde{q}\tilde{r}=\tilde{p}-1=p-1, \\
       &\tilde{q}\tilde{r}^+=(p-1)q.
       \end{split}
   \end{equation}
   Taking $\delta_1=-2\triangle=-2$, (\ref{2.35}) gives rise to
   the stationary relativistic Lotka-Volterra equation
   \begin{equation}\label{2.37}
        \left(
          \begin{array}{c}
          pq-p^-q^-+q^--q\\
            pq-p \\
          \end{array}
        \right)=0,
   \end{equation}
   which means
   \begin{equation}\label{2.38}
     \begin{split}
      &\left(
         \begin{array}{c}
           pq-p^-q^-+q^--q \\
           p^+q^+-pq-p^++p \\
         \end{array}
       \right)\\
      &=\left(
        \begin{array}{cc}
          I & 0 \\
          0 & S^+-I \\
        \end{array}
      \right)
      \left(
        \begin{array}{c}
          pq-p^-q^-+q^--q \\
         pq-p \\
        \end{array}
      \right)\\
      &=0.
      \end{split}
   \end{equation}
   Next we turn to the time-dependent relativistic Lotka-Volterra hierarchy. The zero-curvature equation is
  \begin{equation}\label{2.39}
     U_{t_n}(\lambda)+U(\lambda)V_n(\lambda)-V_n^+(\lambda)U(\lambda)=0.
  \end{equation}
  $U(\lambda), V_n(\lambda),\{a_\ell(\cdot,t)\}_{\ell=0}^{n+1},\{b_\ell(\cdot,t)\}_{\ell=0}^{n},
  \{c_\ell(\cdot,t)\}_{\ell=0}^{n},\{d_\ell(\cdot,t)\}_{\ell=0}^{n+1}$ are defined by
  (\ref{2.1}), (\ref{2.2})-(\ref{2.5}), (\ref{2.11})-(\ref{2.15}) and (\ref{2.21}).
  Then (\ref{2.39}) implies
  \begin{equation}\label{2.40}
      \left(
        \begin{array}{cc}
          -\tilde{p}_{t_n}\lambda & -\tilde{q}_{t_n} \\
          -\tilde{r}_{t_n} & 0 \\
        \end{array}
      \right)+
      \left(
        \begin{array}{cc}
          \Delta_{11}&\Delta_{12}\\
          \Delta_{21}&\Delta_{22} \\
        \end{array}
      \right)=0,
  \end{equation}
  where
  \begin{equation}\label{2.41}
   \Delta_{11}=(\tilde{p}\lambda-\lambda^{-1})A_{2n+1}^-+\tilde{q}C_{2n+1}^--(\tilde{p}\lambda-\lambda^{-1})A_{2n+2}+\tilde{r}B_{2n+1},
  \end{equation}
  \begin{equation}\label{2.42}
    \Delta_{12}=(\tilde{p}\lambda-\lambda^{-1})B_{2n+1}^--\tilde{q}D_{2n+2}^--\tilde{q}A_{2n+2}-B_{2n+1}\lambda,
  \end{equation}
  \begin{equation}\label{2.43}
    \Delta_{21}=\tilde{r}A_{2n+2}^-+C_{2n+1}^-\lambda-(\tilde{p}\lambda-\lambda^{-1})C_{2n+1}+\tilde{r}D_{2n+2}
  \end{equation}
  \begin{equation}\label{2.44}
    \Delta_{22}=\tilde{r}B_{2n+2}^--D_{2n+2}^-\lambda-\tilde{q}C_{2n+1}+D_{2n+2}\lambda.
  \end{equation}
  Inserting (\ref{2.2})-(\ref{2.5}) into (\ref{2.40}) then yields
  \begin{equation}\label{2.45}
    -\tilde{p}_{t_n}+\tilde{p}(a_{n+1}^--a_{n+1})=0,
  \end{equation}
  \begin{equation}\label{2.46}
    -\tilde{q}_{t_n}+b_{n+1}^{-}+\tilde{q}a_{n+1}^--\tilde{q}\triangle=0
  \end{equation}
  \begin{equation}\label{2.47}
    -\tilde{r}_{t_n}+c_{n+1}+\tilde{r}a_{n+1}+\tilde{r}\triangle=0.
  \end{equation}
  Using (\ref{2.23}) (\ref{2.25}) (\ref{2.45}) (\ref{2.46}), we directly calculate
  \begin{equation}\label{2.48}
  \begin{split}
    &\tilde{r}_{t_n}=\partial_{t_n}((\tilde{p-1})/\tilde{q})=\tilde{p}_{t_n}/q-(\tilde{p}-1)\tilde{q}_{t_n}/\tilde{q}^2\\
    &=\tilde{p}(a_{n+1}^--a_{n+1})/\tilde{q}-(\tilde{p}-1)(b_{n+1}^{-}+\tilde{q}a_{n+1}^--\tilde{q}\triangle)/\tilde{q}^2\\
    &=\tilde{p}(a_{n+1}^--a_{n+1})/\tilde{q}-\tilde{r}(\tilde{p}b_{n}^--b_n-\tilde{q}a_{n+1}-\tilde{q}\triangle)/\tilde{q}\\
    &=\tilde{p}(a_{n+1}^--a_{n+1})/\tilde{q}-\tilde{r}\tilde{p}b_{n}^-/\tilde{q}+\tilde{r}b_n/\tilde{q}+(\tilde{r}a_{n+1}+\tilde{r}\triangle)\\
    &=[\tilde{p}(a_{n+1}^--a_{n+1})+\tilde{r}(b_n-\tilde{p}b_n)]/\tilde{q}+(\tilde{r}a_{n+1}+\tilde{r}\triangle)\\
    &=[\tilde{p}(a_{n+1}^--a_{n+1})+\tilde{r}(-b_{n+1}^--\tilde{q}a_{n+1}^--\tilde{q}a_{n+1}]/\tilde{q}+(\tilde{r}a_{n+1}+\tilde{r}\triangle)\\
    &=[\tilde{p}(a_{n+1}^--a_{n+1})+(1-\tilde{p})(a_{n+1}^--a_{n+1})-\tilde{r}b_{n+1}^-]/\tilde{q}+(\tilde{r}a_{n+1}+\tilde{r}\triangle)\\
    &=[a_{n+1}^--a_{n+1}-\tilde{r}b_{n+1}^-]/\tilde{q}+(\tilde{r}a_{n+1}+\tilde{r}\triangle)\\
    &=c_{n+1}+\tilde{r}a_{n+1}+\tilde{r}\triangle.\\
  \end{split}
  \end{equation}
  Thus, (\ref{2.45})-(\ref{2.47}) are essentially two independent equations (\ref{2.45}) and (\ref{2.46}).
  Varying $n\in\mathbb{N}_0$, (\ref{2.45}) and (\ref{2.46}) give rise to time-dependent relativistic
  Lotka-Volterra hierarchy, which we introduce as follows
  \begin{equation}\label{2.49}
      \text{RLV}_n(p,q)=\widetilde{\text{RLV}_n}(\tilde{p},\tilde{q},\tilde{r})
      =\left(
         \begin{array}{c}
           -\tilde{p}_{t_n}+\tilde{p}(a_{n+1}^--a_{n+1})\\
           -\tilde{q}_{t_n}+b_{n+1}^{-}+\tilde{q}a_{n+1}^--\tilde{q}\triangle \\
         \end{array}
       \right)=0,\quad p\in\mathbb{N}_0.
   \end{equation}
   Explicitly,
   \begin{equation}\label{2.50}
       \begin{split}
        &\text{RLV}_0(p,q)=\widetilde{\text{RLV}_0}(\tilde{p},\tilde{q},\tilde{r})\\
        &=\left(\begin{array}{c}
           -\tilde{p}_{t_0}+\tilde{p}(-\tilde{q}\tilde{r}^-+\tilde{q}^+\tilde{r}) \\
           -\tilde{q}_{t_0}-\tilde{p}\tilde{q}+\tilde{q}\tilde{q}^+\tilde{r}+\tilde{q}^+-\tilde{q}(\triangle+\delta_1/2)
         \end{array}\right)\\
         &=\left(
             \begin{array}{c}
               -p_{t_0}+p(pq-p^-q^-+q^--q) \\
               \partial_{t_0}(e^{(S^{+}-I)^{-1}\ln q})+e^{(S^+-I)^{-1}\ln q}\left(pq-p-(\delta_1/2+\triangle)\right) \\
             \end{array}
           \right)\\
         &=\left(
             \begin{array}{c}
               -p_{t_0}+p(pq-p^-q^-+q^--q) \\
               e^{(S^+-I)^{-1}\ln q}(S^{-1}-I)^{-1}(q_{t_n}/q)+e^{(S^+-I)^{-1}\ln q}\left(pq-p-(\delta_1/2+\triangle)\right)\\
             \end{array}\right)\\
         &=0
         \end{split}
  \end{equation}
   where we used
   \begin{equation}\label{2.51}
      \partial_{t_n}(e^{(S^{+}-I)^{-1}\ln q})=e^{(S^+-I)^{-1}}\ln q\partial_{t_n}((S^+-I)^{-1}\ln q)=e^{(S^+-I)^{-1}\ln q}(S^{-1}-I)^{-1}(q_{t_n}/q),
   \end{equation}
   in the third equality of (\ref{2.50}).
   Taking $\delta_1=-2\Delta=-2$, (\ref{2.50}) is just the relativistic Lotka-Volterra equation, which is
   equivalent to the form of
   \begin{equation}\label{2.52}
       \begin{split}
        &\text{RLV}_0(p,q)=\widetilde{\text{RLV}_0}(\tilde{p},\tilde{q},\tilde{r})\\
         &=\left(\begin{array}{c}
            -p_{t_0}+p(pq-p^-q^-+q^--q) \\
           -q_{t_0}+q(p^+q^+-p^+-pq+p)
         \end{array}\right)\\
         &=0.\\
         \end{split}
   \end{equation}
   In order to derive the algebraic curve associated with the relativistic Lotka-Volterra hierarchy,
   we need the some change about $U(\lambda)$, which is
   \begin{equation}\label{2.53}
      U(\lambda)\rightarrow \tilde{U}(\xi)=\xi U(\xi^{-1})=
      \left(
        \begin{array}{cc}
          \tilde{p}-z & \tilde{q}\xi \\
          \tilde{r}\xi & 1 \\
        \end{array}
      \right),\quad z=\xi^2.
      \end{equation}
      Let
      \begin{equation}\label{2.54}
        \tilde{V}_{n}(\xi)=V_{n}(\xi^{-1}),
      \end{equation}
      then $\tilde{U}(\xi), \tilde{V}_{n}(\xi)$ remain satisfy the zero-curvature equation (\ref{2.16}) and (\ref{2.39}), that is
      \begin{equation}\label{2.55}
        \tilde{U}(\xi)\tilde{V}_n(\xi)-\tilde{V}_n^+(\xi)\tilde{U}(\xi)=0
      \end{equation}
      \begin{equation}\label{2.56}
        \tilde{U}_{t_n}(\xi)+\tilde{U}(\xi)\tilde{V}_n(\xi)-\tilde{V}_n^+(\xi)\tilde{U}(\xi)=0
      \end{equation}
      and hence the final form of (\ref{2.34}) (\ref{2.49}) is invariant.
      Assume
      \begin{equation}\label{2.57}
        \tilde{V}_{n}(\xi)=V_{n}(\xi^{-1})
        =\left(
        \begin{array}{cc}
        \tilde{A}_{n+1}(z) & \xi\tilde{B}_{n}(z) \\
        \xi\tilde{C}_{n}(z) & -\tilde{D}_{n+1}(z) \\
        \end{array}
        \right)
      \end{equation}
      and then one derives
      \begin{equation}\label{2.58}
        \tilde{A}_{n+1}(z)=A_{2n+2}(\xi^{-1}),
      \end{equation}
      \begin{equation}\label{2.59}
        \tilde{B}_{n}(z)=\xi^{-1}B_{2n+1}(\xi^{-1}),
      \end{equation}
      \begin{equation}\label{2.60}
        \tilde{C}_{n}(z)=\xi^{-1}C_{2n+1}(\xi^{-1}),
      \end{equation}
      \begin{equation}\label{2.61}
        \tilde{D}_{n+1}(z)=D_{2n+2}(\xi^{-1})
      \end{equation}
      by comparing with (\ref{2.11}). Then
      (\ref{2.17})-(\ref{2.10}) change into
      \begin{equation}\label{2.17a}
     (\tilde{p}-z)\tilde{A}_{n+1}^-+z\tilde{q}\tilde{C}_n^{-}=
     (\tilde{p}-z)\tilde{A}_{n+1}+z\tilde{r}\tilde{B}_{n+1},
   \end{equation}
   \begin{equation}\label{2.18a}
     (\tilde{p}-z)\tilde{B}_{n}^--\tilde{q}\tilde{D}_{n+1}^-=\tilde{q}\tilde{A}_{n+1}+\tilde{B}_{n},
   \end{equation}
   \begin{equation}\label{2.19a}
     \tilde{r}\tilde{A}_{n+1}^-+\tilde{C}_{n}^-=(\tilde{p}-z)\tilde{C}_{n}-\tilde{r}\tilde{D}_{n+1},
   \end{equation}
   \begin{equation}\label{2.20a}
     z\tilde{r}\tilde{B}_{n}^--D_{n+1}^-=z\tilde{q}\tilde{C}_{n}-\tilde{D}_{n+1},
   \end{equation}
   respectively.

      Taking into account (\ref{2.55}) (\ref{2021}), one infers that the expression $R_{2n+2}(z)$, defined as
      \begin{equation}\label{2.62}
         R_{2n+2}(z)=-(\tilde{A}_{n+1})^2-z\tilde{B}_n\tilde{C}_n,
      \end{equation}
      is a lattice constant, that is, $R_{2n+2}=R_{2n+2}^-$, since taking determinants in the stationary
      zero-curvature equation (\ref{2.55}) immediately yields
      \begin{equation}\label{2.63}
        \tilde{p}(1-z)[-(\tilde{A}_{n+1})^2-\tilde{B}_n^-\tilde{C}_n^-+(\tilde{A}_{n+1})^2+\tilde{B}_{n}\tilde{C}_n]=0.
      \end{equation}
      Hence, $R_{2n+2}$ only depends on $z$, and one may write $R_{2n+2}$ as
      \begin{equation}\label{2.64}
        R_{2n+2}(z)=-(1/4)\prod_{m=0}^{2n+1}(z-E_m),\quad n\in\mathbb{N}_0.
      \end{equation}

\begin{remark}\label{remark3}
(i) Taking the transformation $p\rightarrow hp, q\rightarrow hq,$ one concludes that the equation (\ref{2.34}) and (\ref{2.49}) change into the normal form of relativistic Lotka-Volterra system \cite{RLV,Laxpresentation}.

(ii) Taking the transformation $\tilde{p}=p=\tilde{v}^+/\tilde{v},\tilde{v}\in\ell(\mathbb{Z})$, (\ref{2.34}) and (\ref{2.49}) change into
\begin{equation}\label{2.65}
    \text{s-}\overline{\widetilde{\text{RLV}}}_n(\tilde{v},\tilde{q},\tilde{r})=\text{s-}\widetilde{\text{RLV}_n}(\tilde{v}^+/\tilde{v},\tilde{q},\tilde{r})=0,
\end{equation}
\begin{equation}\label{2.66}
   \overline{\widetilde{\text{RLV}}}_n(\tilde{v},\tilde{q},\tilde{r})=\widetilde{\text{RLV}_n}(\tilde{v}^+/\tilde{v},\tilde{q},\tilde{r})=0.
\end{equation}

\end{remark}
\section{Algebro-geometric Solutions of Stationary Relativistic Lotka-Volterra Hierarchy}

In this section, we present a detailed study of the stationary Toda hierarchy.
Our principle tools are derived from the polynomial recursion formalism introduced in section 2 and
a fundamental meromorphic function $\phi$
on $\mathcal{K}_n$. With the help of $\phi$ we study the Baker-Akhiezer vector $\Psi$, the common eigenfunction\
of $\tilde{U}(\xi)$ and $\tilde{V}_n(\xi)$, trace formulas, and theta function representations of
$\phi$, $\Psi$, $p$ and $q$.

Throughout this paper we have the following hypothesis
\begin{guess}
     \text{Assume}
      \begin{equation}\label{3.1}
             \tilde{p}=p\neq 0,1,\quad\tilde{p}(n),\tilde{q}(n),\tilde{r}(n),p(n),q(n)\in\ell(\mathbb{Z}),\quad n\in\mathbb{Z},
     \end{equation}
and $n\in\mathbb{N}_0$ in (\ref{2.34}) fixed.
\end{guess}

\begin{guess}
    The affine part of the algebraic curve associated with the $n$-th equation in the stationary relativistic
    Lotka-Volterra hierarchy (\ref{2.34}), which takes the form of
    \begin{equation}\label{3.2}
       \mathcal{K}_n:\mathcal{F}_n(z,y)= y^2+4R_{2n+2}(z)=y^2-\prod_{m=0}^{2n+1}(z-E_m)=0,\quad\{E_m\}_{m=0}^{2n+1}\subseteq\mathbb{C}\backslash\{0\}
    \end{equation}
    is nonsingular.
     That is,
    \begin{equation}\label{3.3}
      \begin{split}
        &E_m\neq E_{m^{'}}\quad for\quad m\neq m^{'}, m, m^{'}=0,1,2,\dotsi,2n+1.\\
        \end{split}
    \end{equation}
\end{guess}
$\mathcal{K}_{n}$ defined in (\ref{3.2}) is compactified by joining two points $P_{\infty\pm}, P_{\infty+}\neq P_{\infty-}$ at infinity, but for notational simplicity the compactification is also denoted by $\mathcal{K}_n$.

One can introduce the complex structure on $\mathcal{K}_{n}$ to yield a compact Riemann surface of genus $n$ \cite{r,s,16}. This is a hyperelliptic Riemann surface. It can be regarded as a double covering of sphere surface $\mathbb{C}_{\infty}=\mathbb{C}\bigcup\{\infty\}$.
Points $P$ on $\mathcal{K}_p\backslash\{P_{\infty\pm}\}$ are represented as pairs $P=(z,y)$, where $y(\cdot)$ is the meromorphic function on $\mathcal{K}_n$ satisfying $\mathcal{F}_{n}(z,y)=0$.

We write
\begin{equation}\label{3.4}
    \tilde{B}_{n}(z)=(-\tilde{q}^{+})\prod_{j=1}^{n}(z-\mu_j),\quad \tilde{C}_{n}(z)=(-\tilde{r})\prod_{j=1}^{n}(z-\nu_j),\quad\mu_j,\nu_j\in\ell(\mathbb{Z}).
\end{equation}
Next we ¡¯lift¡¯ the point $\mu_j,$
$\nu_j$ from the $\mathbb{C}_{\infty}$ to the compact Riemann surface $\mathcal{K}_n,$
\begin{equation}\label{3.5}
    \hat{\mu}_j=(\mu_j,-2\tilde{A}_{n+1}(\mu_j,n)),\quad \hat{\nu}_j=(\nu_j,2\tilde{A}_{n+1}(\nu_j,n)),
    \quad j=1,2,\dotsi,n
\end{equation}
and introduce the points $P_{0,\pm}$ on $\mathcal{K}_n$
\begin{equation}\label{3.6}
    P_{0,\pm}=(0,\pm2\tilde{A}_{n+1}(0,n)),\quad \tilde{A}_{n+1}^2(0,n)=\frac{1}{4}\prod_{m=0}^{2n+1}E_m.
\end{equation}
We introduce the holomorphic sheet exchange map on $\mathcal{K}_n$
$$*:\quad \mathcal{K}_n\rightarrow\mathcal{K}_n\quad P=(z,y)\mapsto P^{*}=(z,-y)\quad P_{\infty\pm}\mapsto P_{\infty\pm}^{*}=P_{\infty\mp}.$$

Next one can define a fundamental meromorphic
function $\tilde{\phi}$ on $\mathcal{K}_n$
\begin{equation}\label{3.7}
    \phi(P,n)=\frac{\frac{1}{2}y-\tilde{A}_{n+1}(z,n)}{B_{n}(z,n)}
    =\frac{z\tilde{C}_{n}(z,n)}{\frac{1}{2}y+\tilde{A}_{n+1}(z,n)},
    \quad P=(z,y)\in\mathcal{K}_n.
\end{equation}
The divisor $(\tilde{\phi}(\cdot))$ of $\tilde{\phi}$ is
\begin{equation}\label{3.8}
    (\phi(\cdot,n))=\mathcal{D}_{P_{0,+}\underline{\hat{\nu}}(n)}-\mathcal{D}_{P_{\infty+}\underline{\hat{\mu}}(n)},
\end{equation}
where we abbreviated
\begin{equation}\label{3.9}
    \underline{\hat{\mu}}(n)=\left(\mu_1(n),\mu_2(n),\dotsi,\mu_n(n)\right),\quad \underline{\hat{\nu}}(n)=\left(\nu_1(n),\nu_2(n),\dotsi,\nu_n(n)\right).
\end{equation}
Given $\phi(\cdot,n)$, the stationary Baker-Akhiezer vector $\Psi(\cdot,n,n_0)$ on
$\mathcal{K}_n$ is then defined by
\begin{equation}\label{3.10}
    \Psi(P,\xi,n,n_0)=\left(
                 \begin{array}{c}
                    \Psi_1(P,\xi,n,n_0) \\
                    \Psi_2(P,\xi,n,n_0)
                  \end{array}
                  \right),
\end{equation}

\begin{equation}\label{3.11}
    \Psi_1(P,\xi,n,n_0)=
    \begin{cases}
      \prod_{n^{'}=n_0+1}^{n}\left(\tilde{p}(n^{'})-z+\tilde{q}(n^{'})\tilde{\phi}^-(P,n^{'})\right),&n\geq n_0+1,\cr
      1,&n=n_0,\cr
      \prod_{n^{'}=n+1}^{n_0}\left(\tilde{p}(n^{'})-z+\tilde{q}(n^{'})\tilde{\phi}^-(P,n^{'})\right)^{-1},&n\leq n_0-1,
    \end{cases}
\end{equation}

\begin{equation}\label{3.12}
    \Psi_2(P,\xi,n,n_0)=\xi^{-1}\tilde{\phi}(P,n_0)
    \begin{cases}
      \prod_{n^{'}=n_0+1}^{n}\left(\frac{\tilde{r}(n^{'})z}{\tilde{\phi}^-(P,n^{'})}+1\right),&n\geq n_0+1,\cr
      1,&n=n_0,\cr
      \prod_{n^{'}=n+1}^{n_0}\left(\frac{\tilde{r}(n^{'})z}{\tilde{\phi}^-(P,n^{'})}+1\right)^{-1},&n\leq n_0-1.
    \end{cases}
\end{equation}

Some properties of $\tilde{\phi}, \Psi_1, \Psi_2$ are discussed in the following lemma.
\begin{lemma}\label{lemma1}
    Suppose $p,q$ satisfy the $n$-th relativistic Lotka-Volterra hierarchy (\ref{2.34}) and (\ref{3.1})-(\ref{3.12}) holds. Let $P=(z,y)\in\mathcal{K}_n\backslash\{P_{\infty\pm},P_{0,\pm}\}, (n,n_0)\in\mathbb{Z}^2.$ Then $\tilde{\phi}$ satisfies the Riccati-type equation
    \begin{equation}\label{3.13}
        (\tilde{p}-z)\tilde{\phi}+\tilde{q}\tilde{\phi}\tilde{\phi}^-=\tilde{r}z+\tilde{\phi}^{-},
    \end{equation}
    as well as
    \begin{equation}\label{3.14}
        \tilde{\phi}(P)\tilde{\phi}(P^{*})=\frac{-z\tilde{C}_{n}(z)}{\tilde{B}_{n}(z)},
    \end{equation}
    \begin{equation}\label{3.15}
        \tilde{\phi}(P)+\tilde{\phi}(P^{*})=\frac{-2\tilde{A}_{n+1}}{\tilde{B}_n},
   \end{equation}
   \begin{equation}\label{3.16}
        \tilde{\phi}(P)-\tilde{\phi}(P^{*})=\frac{y}{\tilde{B}_n(z)}.
   \end{equation}
   Moreover, the vector $\Psi$ satisfy
   \begin{equation}\label{3.17}
     \tilde{U}(\xi)\Psi^-(P)=\Psi(P),
   \end{equation}
   \begin{equation}\label{3.18}
     \tilde{V}_n(\xi)\Psi(P)=(1/2)y\Psi(P),
   \end{equation}
    \begin{equation}\label{3.19}
        \Psi_2(P,\xi,n,n_0)=\xi^{-1}\tilde{\phi}(P,n)\Psi_1(P,\xi,n,n_0).
    \end{equation}

\end{lemma}
\begin{proof}
Using (\ref{2021}) (\ref{2.18a}) (\ref{2.20a}) (\ref{2.62}) (\ref{2.64}) (\ref{3.7}), we directly calculate

   \begin{equation}\label{3.20}
        \begin{split}
        &(\tilde{p}-z)\tilde{\phi}+\tilde{q}\tilde{\phi}\tilde{\phi}^--\tilde{r}z-\tilde{\phi}^{-}\\
        &=1/(\tilde{B}_n\tilde{B}_n^-)[\tilde{B}_n^-(\tilde{p}-z)(\frac{1}{2}y-\tilde{A}_{n+1})+\tilde{q}(\frac{1}{2}y-\tilde{A}_{n+1})(\frac{1}{2}-\tilde{A}_{n+1}^-)\\
        &-\tilde{r}\tilde{B}_n\tilde{B}_n^--(\frac{1}{2}y-\tilde{A}_{n+1})\tilde{B}_n]\\
        &=0.
        \end{split}
    \end{equation}
Equalities (\ref{3.14})-(\ref{3.16}) hold from (\ref{3.7}).
Next we use induction to prove (\ref{3.19}).
Obviously, for $n=n_0$, we have
    \begin{equation}\label{3.21}
           \Psi_2(P,\xi,n_0,n_0)=\xi^{-1}\tilde{\phi}(P,n_0)\Psi_1(P,\xi,n_0,n_0)
    \end{equation}
from (\ref{3.11}) and (\ref{3.12}).
In the case $k>n_0$, assume (\ref{3.19}) holds for $n=n_0,n_0+1,\dotsi,k-1$.
From the definition of $\Psi_1$, $\Psi_2$ in (\ref{3.11}) (\ref{3.12}), one finds
\begin{equation}\label{3.22}
  \begin{split}
  &\xi\Psi_2(P,\xi,k,n_0)/\Psi_1(P,\xi,k,n_0)=\xi\times[\Psi_2(P,\xi,k-1,n_0)/\Psi_1(P,\xi,k-1,n_0)]\\
  &\times[\left(\frac{\tilde{r}z+1}{\tilde{\phi}^-(P,k)}\right)/\left(\tilde{p}-z+\tilde{q}\tilde{\phi}^-(P,k)\right)]\\
  &=\xi\times\xi^{-1}\tilde{\phi}(P,k-1)\times[\left(\frac{\tilde{r}z+1}{\tilde{\phi}^-(P,k)}\right)/\left(\tilde{p}-z+\tilde{q}\tilde{\phi}^-(P,k)\right)]\\
  &=\left(\tilde{r}z+1\right)/\left(\tilde{p}-z+\tilde{q}\tilde{\phi}^-(P,k)\right).
\end{split}
\end{equation}
The Riccati-type equation (\ref{3.13}) shows that $\tilde{\phi}(P,k)$ satisfies
\begin{equation}\label{3.23}
  \tilde{\phi}(P,k)=\left(\tilde{r}z+1\right)/\left(\tilde{p}-z+\tilde{q}\tilde{\phi}^-(P,k)\right).
\end{equation}
Noted that $\xi\Psi_2(P,\xi,k,n_0)/\Psi_1(P,\xi,k,n_0)$ and $\tilde{\phi}(P,k)$ take the same value at $n=n_0$
and (\ref{3.22}) (\ref{3.23}), we have (\ref{3.19}) for $n=k$.
The proof of the case $n<n_0$ is similar with $n>n_0$.
From (\ref{3.11}) (\ref{3.12}), one  finds
\begin{equation}\label{3.24}
\begin{split}
  &\Psi_1=\left(\tilde{p}-z+\tilde{q}\tilde{\phi}^-\right)\Psi_1^-\\
  &=(\tilde{p}-z)\Psi_1^-+\xi\tilde{q}\Psi_2^-\\
\end{split}
\end{equation}
and
\begin{equation}\label{3.25}
\begin{split}
  &\Psi_2=\left(\frac{\tilde{r}z}{\tilde{\phi}^-}+1\right)\Psi_2^-\\
  &=\xi\tilde{r}\Psi_1^-+\Psi_2^-,\\
\end{split}
\end{equation}
where we used (\ref{3.19}) in the last equalities of (\ref{3.24}) and (\ref{3.25}).
(\ref{3.18}) comes from (\ref{3.7}) and (\ref{3.19}) by direct calculation.

\end{proof}

Combining the polynomial recursive relation defined in section 1 with (\ref{3.4}) then yields the following
trace formula for $a_\ell$, and $ b_\ell$ in terms of symmetric functions of $\mu_j$ and $\nu_j$, respectively.
\begin{lemma}\label{lemma2}
Suppose $p, q$ satisfy the $n$-th stationary relativistic Lotka-Volterra system (\ref{2.34}) and (\ref{3.1}).
Then we have the following trace formula
\begin{equation}\label{3.26}
  \tilde{p}^{+}-\tilde{q}^{+}\tilde{r}-\tilde{q}^{++}\tilde{r}^+-\tilde{q}^{++}/\tilde{q}^{+}+\delta_1=-\sum_{j=1}^{n}\mu_j,
\end{equation}
\begin{equation}\label{3.27}
  \tilde{p}-\tilde{q}^{+}\tilde{r}-\tilde{r}^{-}/\tilde{r}-\tilde{q}\tilde{r}^{-}+\delta_1=-\sum_{j=1}^{n}\nu_j.
\end{equation}

\end{lemma}
\begin{proof}
   Comparing the coefficients of $z^{n-1}$ of $\tilde{B}_n(z), \tilde{C}_n(z)$ in (\ref{2.59}) (\ref{2.60}) and (\ref{3.4})
the yield the above formulas (\ref{3.26}) (\ref{3.27}).

\end{proof}

Next we turn to study the asymptotic behavior of $\tilde{\phi}, \Psi_1, \Psi_2$ in the
neigborhood of $P_{\infty\pm}$ and $P_{0,\pm}$. This is a crucial step to construct the
algebro-geometric solutions of relativistic Lotka-Volterra hierarchy.

\begin{lemma}\label{lemma3}
 Suppose $p, q$ satisfy the $n$-th stationary relativistic Lotka-Volterra system (\ref{2.34}) and (\ref{3.1}).
Moreover, let $P=(z,y)\in\mathcal{K}_n\backslash\{P_{\infty\pm}, P_{0,\pm}\}, (n,n_0)\in\mathbb{Z}\times\mathbb{Z}.$
Then $\tilde{\phi}$ defined in (\ref{3.7}) has the following asymptotic property
\begin{equation}\label{3.28}
  \tilde{\phi}(P)=
  \begin{cases}
      (\tilde{q}^{+})^{-1}\zeta^{-1}+\left(((\tilde{q}^{+}/\tilde{q}-1)\tilde{p})/\tilde{q}\right)^{+}+O(\zeta)
                                                                             &\text{as}\quad P\rightarrow P_{\infty+}, \cr
      -\tilde{r}+(\tilde{p}\tilde{r}^--\tilde{p}\tilde{r})\zeta+O(\zeta^2)
                                                                              &\text{as}\quad P\rightarrow P_{\infty-}, \cr
  \end{cases}
\end{equation}
where we use the local coordinate $z=\zeta^{-1}$ near the points $P_{\infty\pm}$.
\begin{equation}\label{3.29}
    \tilde{\phi}(P)=
  \begin{cases}
     c_{n}/(\prod_{m=0}^{2n+1}E_{m})\zeta+O(\zeta^2)
                                                                             &\text{as}\quad P\rightarrow P_{0,+}, \cr
      -\left(\prod_{m=0}^{2n+1}E_m\right)/b_n+O(\zeta)
                                                                              &\text{as}\quad P\rightarrow P_{0,-}, \cr
  \end{cases}
\end{equation}
 where we use the local coordinate $z=\zeta$ near the points $P_{0,\pm}.$

The components $\Psi_1, \Psi_2$ of the Baker-Akhiezer $\Psi$ have the following asymptotic
properties
\begin{equation}\label{3.30}
    \Psi_1(P,\xi,n,n_0)=
  \begin{cases}
     \left(\tilde{q}^{+}(n)\tilde{v}^{+}(n)\right)/\left(\tilde{q}^{+}(n_0)\tilde{v}^{+}(n_0)\right)+O(\zeta)
                                                                             &\text{as}\quad P\rightarrow P_{\infty+}, \cr
     (-1)^{n-n_0}\times\zeta^{n_0-n}\left(1+O(\zeta)\right)
                                                                              &\text{as}\quad P\rightarrow P_{\infty-},\cr
  \end{cases}
\end{equation}

\begin{equation}\label{3.31}
    \Psi_1(P,\xi,n,n_0)=
  \begin{cases}
     \tilde{v}^{+}(n)/\tilde{v}^{+}(n_0)+O(\zeta)
                                                                             &\text{as}\quad P\rightarrow P_{0,+}, \cr
      \Gamma(\tilde{p}-\tilde{q}\frac{\prod_{m=0}^{2n+1}E_m}{b_n^-})(n,n_0)+O(\zeta)
                                                                              &\text{as}\quad P\rightarrow P_{0,-}, \cr
  \end{cases}
\end{equation}
and
\begin{equation}\label{3.32}
  \Psi_2(P,\xi,n,n_0)=\xi^{-1}\times
  \begin{cases}
     \tilde{v}^{+}(n)/\left(\tilde{q}^{+}(n_0)\tilde{v}^{+}(n_0)\right)\zeta^{-1}\left(1+O(\zeta)\right)
                                                                             &\text{as}\quad P\rightarrow P_{\infty+}, \cr
      (-1)^{n+1-n_0}\times\tilde{r}(n)\zeta^{n_0-n}\left(1+O(\zeta)\right)
                                                                              &\text{as}\quad P\rightarrow P_{\infty-}, \cr
  \end{cases}
\end{equation}
\begin{equation}\label{3.33}
    \Psi_2(P,\xi,n,n_0)=\xi^{-1}\times
  \begin{cases}
    [\left(\tilde{v}^{+}(n)c_n\right)/\left(\prod_{m=0}^{2n+1}E_m\tilde{v}^{+}(n_0)\right)]\zeta+O(\zeta^2)\cr
                    \qquad\qquad\qquad\qquad\qquad\qquad\qquad\qquad\quad\text{as}\quad                                                         P\rightarrow P_{0,+},\cr
      -\left(\prod_{m=0}^{2n+1}E_m\left(\Gamma(\tilde{p}-\tilde{q}\frac{\prod_{m=0}^{2n+1}E_m}{b_n^-})(n,n_0)\right)\right)/b_n +O(\zeta)\\
                    \qquad\qquad\qquad\qquad\qquad\qquad\qquad\qquad\quad\text{as}\quad                                                         P\rightarrow P_{0,-}, \cr
  \end{cases}
\end{equation}
where $\tilde{v}$ are defined in Remark \ref{remark3} and
   \begin{equation}\label{3.34}
      \Gamma(f)(n,n_0)=
           \begin{cases}
                  \prod_{n^{'}=n_0+1}^{n}f(n^{'})& n>n_0,\cr
                  1&n=n_0,\cr
                  \prod_{n^{'}=n+1}^{n_0}f(n^{'})^{-1}& n<n_0,\cr
           \end{cases}\quad  \forall f\in\ell(\mathbb{Z}).
   \end{equation}

\end{lemma}
 \begin{proof}
   The existence of the asymptotic expansion of $\tilde{\phi}$ in terms of the local coordinate
   $z=\zeta^{-1}$ near $P_{\infty\pm}$, respectively, $\zeta=z$ near $P_{0,\pm}$ is clear from the
   explicit form of $\tilde{\phi}$ in (\ref{3.7}) (\ref{3.8}).
   Assume $\tilde{\phi}$ has the following asymptotic expansions
\begin{equation}\label{3.35}
\tilde{\phi}=
\begin{cases}
 \phi_{-1}\zeta^{-1}+\phi_{0}+\phi_{1}\zeta+O(\zeta^2),&\text{as}\quad P\rightarrow P_{\infty+},\cr
 \phi_{0}+\phi_{1}\zeta+\phi_2\zeta^2+O(\zeta^3), &\text{as}\quad P\rightarrow P_{\infty-}.\cr
\end{cases}
\end{equation}
Inserting the asymptotic expansions (\ref{3.35}) into the Riccati-type equation (\ref{3.13}) and comparing
coefficients of powers of $\zeta$, which determines the coefficients $\phi_k$ in (\ref{3.25}), one concludes
(\ref{3.28}). Insertion of the polynomials $\tilde{B}_{n}, \tilde{C}_{n}$ defined in (\ref{2.59}) (\ref{2.60})
into (\ref{3.7}) then yields the explicit coefficients in (\ref{3.29}).
Next we compute the asymptotic expansion of $\Psi_1$. Noted the definition of $\Psi_1$ in (\ref{3.11}), we first investigate the expression  $\tilde{p}-z+\tilde{q}\tilde{\phi}^-$. With the help of (\ref{3.28}) (\ref{3.29}),
one finds
\begin{equation}\label{3.36}
  \tilde{p}-z+\tilde{q}\tilde{\phi}^-=
  \begin{cases}
 \left(\tilde{q}^{+}\tilde{v}^{+}\right)/\left(\tilde{q}\tilde{v}\right)+O(\zeta),&\text{as}\quad P\rightarrow P_{\infty+},\cr
  -\zeta^{-1}+O(1),&\text{as}\quad P\rightarrow P_{\infty-},\cr
  \tilde{v}^{+}/\tilde{v}+O(\zeta),&\text{as}\quad P\rightarrow P_{0,+},\cr
  \tilde{p}-\tilde{q}\frac{\prod_{m=0}^{2n+1}E_m}{b_n^-}+O(\zeta),&\text{as}\quad P\rightarrow P_{0,-},\cr
  \end{cases}
\end{equation}
which give rise to (\ref{3.30}) (\ref{3.31}).
Obviously, the exact $n$ poles of $\tilde{p}-z+\tilde{q}\tilde{\phi}^-$ in $\mathcal{K}_n\backslash\{P_{\infty\pm},P_{0,\pm}\}$ coincide with the ones of $\tilde{\phi}^-$ and there are $n+1$
poles of $\tilde{p}-z+\tilde{q}\tilde{\phi}^-$ in $\mathcal{K}_n$. However, it is easy to know that $\Psi_1$ is a meromorphic function on $\mathcal{K}_n$ from the definition (\ref{3.11}) and the meromorphic property of $\tilde{\phi}.$ Therefore the meromorphic function $\tilde{p}-z+\tilde{q}\tilde{\phi}^-$ possess exact $n+1$ zero points on $\mathcal{K}_n.$
So we need some deformation of function $\tilde{p}-z+\tilde{q}\tilde{\phi}^-$.
Using (\ref{2.18a}) (\ref{2.62}) (\ref{3.2}) and (\ref{3.7}), one may calculate
\begin{equation}\label{3.37}
\begin{split}
&\tilde{p}-z+\tilde{q}\tilde{\phi}^-\\
&=\tilde{p}-z+\tilde{q}\frac{\frac{1}{2}y-\tilde{A}_{n+1}^{-}}{\tilde{B}_{n}^{-}}\\
&=\frac{(\tilde{p}-z)\tilde{B}_{n}^{-}+\tilde{q}(\frac{1}{2}y-\tilde{A}_{n+1}^-)}{\tilde{B}_{n}^{-}}\\
&=\frac{\tilde{q}\tilde{A}_{n+1}+\tilde{B}_n+\frac{1}{2}\tilde{q}y}{\tilde{B}_n^{-}}\\
&=\frac{\tilde{B}_n}{\tilde{B}_{n}^{-}}+\tilde{q}\frac{\frac{1}{2}y+\tilde{A}_{n+1}}{\tilde{B}_{n}^{-}}\\
&=\frac{\tilde{B}_n}{\tilde{B}_{n}^{-}}+
\tilde{q}\frac{\frac{1}{4}y^2-\tilde{A}_{n+1}^2}{\tilde{B}_{n}^{-}\left(\frac{1}{2}y-\tilde{A}_{n+1}\right)}\\
&=\frac{\tilde{B}_n}{\tilde{B}_{n}^{-}}\left(1+\tilde{q}\frac{z\tilde{C}_{n}}{\frac{1}{2}y-\tilde{A}_{n+1}}\right).\\
\end{split}
\end{equation}
Thus,
\begin{equation}\label{3.38}
   \tilde{p}-z+\tilde{q}\tilde{\phi}^-(P)\equfill{P\rightarrow \hat{\mu}_j}{}\frac{\tilde{B}_{n}(P)}{\tilde{B}_{n}^-(P)}O(1),
\end{equation}
which shows $\hat{\mu}_j, \quad j=1,\dotsi,n$ are $n$ zero points of function $\tilde{p}-z+\tilde{q}\tilde{\phi}^-(P).$  The remaining one zero point is denoted by $P_{\sharp}$.
Finally, (\ref{3.32}) and (\ref{3.33})
follows from (\ref{3.19}) and (\ref{3.28})-(\ref{3.31}).

\end{proof}
We choose a fixed base point $Q_{0}$ on $\mathcal{K}_{p}\backslash\{P_{0,+,\pm},P_{\infty\pm}\}$. Let $\omega_{P_{0}P_{\infty+}}^{(3)}$ be a normal
differential of the third kind holomorphic on $\mathcal{K}_{p}\backslash\{P_{\infty+},P_{0,+}\}$ with simple poles at
$P_{\infty+}$ and $P_{0}$ and residues -1 and 1, respectively, that is,
\begin{eqnarray}
\omega_{P_{0,+}P_{\infty+}}^{(3)}\equfill{\zeta\rightarrow 0}{}
\begin{cases}
(-\zeta^{-1}+O(1))d\zeta&P\rightarrow P_{\infty+}\cr
(\zeta^{-1}+O(1))d\zeta&P\rightarrow P_{0,+}
\end{cases}
\end{eqnarray}
and
\begin{equation}\label{3.321}
\omega_{P_{\infty+}P_{\infty-}}^{(3)}=\frac{1}{y}\prod_{j=1}^p\left(z-\lambda_j^{'}\right)dz
\end{equation}
be a normal
differential of the third kind holomorphic on $\mathcal{K}_{p}\backslash\{P_{\infty+},P_{\infty-}\}$ with simple poles at
$P_{\infty+}$ and $P_{\infty-}$ and residues 1 and -1,
where the local coordinates
$z=\zeta^{-1}$ for $P$ near $P_{\infty\pm}$ , $z=\zeta$ for $P$ near $P_{0}$, and $\{\lambda_{j}\}_{j=1,\dotsi,p}$, $\{\lambda_{j}^{'}\}_{j=1,\dotsi,p}$ are constants uniquely determined by normalized process.

Moreover,
\begin{equation}
\int_{a_{j}}\omega_{P_{0,+}P_{\infty+}}^{(3)}=0,\quad j=1,\dotsi,p,
\end{equation}
\begin{equation}
\int_{Q_{0}}^{P}\omega_{P_{0,+}P_{\infty+}}^{(3)}\equfill{\zeta\rightarrow 0}{}\left(\begin{array}{cccc}0\\\ln\zeta\end{array}\right)+\left(\begin{array}{cccc}e_{0.-}\\e_{0,+}\end{array}\right)+O(\zeta)\quad\begin{array}{cccc}P\rightarrow P_{\infty-}\\P\rightarrow P_{\infty+}\end{array},
\end{equation}
\begin{equation}
\int_{Q_{0}}^{P}\omega_{P_{0,+}P_{\infty+}}^{(3)}\equfill{\zeta\rightarrow 0}{}-\ln\zeta+d_{0}+O(\zeta)\quad P\rightarrow P_{0},
\end{equation}
where we choose a homology basis $\{a_{j},b_{j}\}_{j=1}^{n}$ on $\mathcal{K}_{n}$ in such a way that the intersection matrix of the cycles satisfies
\begin{equation}
a_{j}\circ b_{k}=\delta_{j,k}, \quad a_{j}\circ a_{k}=0,\quad b_{j}\circ b_{k}=0,\quad j,k=1,\dotsi,n
\end{equation}
and  $e_{0,\pm},d_{0}\in\mathbb{C}.$
One easily verifies that $dz/y$ is a differential on $\mathcal{K}_{n}$ with zeros of order $p-1$ at $P_{\infty\pm}$ and hence $$\eta_{j}=\frac{z^{j-1}}{y}dz,\quad j=1,\dotsi,n$$ form a basis for the space of holomorphic differentials on $\mathcal{K}_{n}$. Introducing the following invertible matrix $C_{j,k}\in\mathbb{C}$
\begin{equation}
C=(C_{j,k})_{j,k=1,\dotsi,p},\quad C_{j,k}=\int_{a_{k}}\eta_{j},
\end{equation}
\begin{equation}
\underline{c}(k)=(c_{1}(k),\dotsi,c_{p}(k)),\quad c_{j}(k)=(C^{-1})_{j,k} \quad k=1,\dotsi,n.
\end{equation}
It's easy to show that the normalized holomorphic differentials $\{\omega_{j}\}_{j=1,\dotsi,p}$ can be written into
\begin{equation}
\omega_{j}=\sum_{l=1}^{p}c_{j}(l)\eta_{l},\quad \int_{a_{k}}\eta_{j}=\delta_{j,k},\quad j,k=1,\dotsi,n,
\end{equation}
\begin{equation}
\underline{\omega}=\left(\omega_1,\dotsi,\omega_n\right).
\end{equation}
Assume $\eta\in\mathbb{C}$ and $|\eta|<$min$\{|E_0|^{-1},|E_{1}|^{-1},|E_{2}|^{-1},\dotsi,|E_{2n+1}|^{-1}\}$
and abbreviate $$\underline{E}=(E_{0},E_{1},\dotsi,E_{2n+1}).$$ Then $$\left(\prod_{m=0}^{2n+1}(1-E_{m}\eta)\right)^{-1/2}=\sum_{k=0}^{+\infty}\hat{c}_{k}(\underline{E})\eta^{k},$$
where $$\hat{c}_{0}(\underline{E})=1,\quad \hat{c}_{1}(\underline{E})=\frac{1}{2}\sum_{m=0}^{2n+1}E_{m}, \quad etc.$$
Similarly,
$$\left(\prod_{m=0}^{2n+1}(1-E_{m}\eta)\right)^{1/2}=\sum_{k=0}^{+\infty}{c}_{k}(\underline{E})\eta^{k},$$
where $${c}_{0}(\underline{E})=1,\quad {c}_{1}(\underline{E})=-\frac{1}{2}\sum_{m=0}^{2n+1}E_{m}, \quad etc.$$
Obviously,
\begin{equation}
\begin{split}
&y(P)=\mp\zeta^{-n-1}\sum_{k=0}^{+\infty}c_{k}(\underline{E})\zeta^{k}\\
&=\mp\left(1-\frac{1}{2}\left(\sum_{m=0}^{2n+1}E_m\zeta+O(\zeta^{2})\right)\right)
\quad\text{as}\quad P\rightarrow P_{\infty\pm},\quad z=\zeta^{-1}.\\
\end{split}
\end{equation}
In the following it will be convenient to introduce the abbreviations
\begin{equation}
\underline{z}(P,\underline{Q})=\underline{\Xi}_{Q_{0}}-\underline{A}_{Q_{0}}(P)+\underline{\alpha}_{Q_{0}}(\mathcal{D}_{\underline{Q}}),\quad P\in\mathcal{K}_{p},\quad\underline{Q}=\{Q_{1},\dotsi,Q_{p}\}\in\text{Sym}^{n}(\mathcal{K}_{p}),
\end{equation}
where $\underline{\Xi}_{Q_{0}}$ is the vector of Riemann constants and the Abel maps $\underline{A}_{Q_0}(\cdot), \underline{\alpha}_{Q_0}(\cdot)$
are defined by (period lattice $L_p=\{\underline{z}\in\mathbb{Z}^g|\underline{z}=\underline{n}+\underline{m}\tau, \underline{n},\underline{m}\in\mathbb{Z}^g\}$)
\begin{equation}
\begin{split}
&\underline{A}_{Q_0}: \mathcal{K}_p\rightarrow \mathcal{J}(\mathcal{K}_p)=\mathbb{Z}^p/L_p\\
&P\mapsto\underline{A}_{Q_0}(P)=\left(A_{Q_0,1}(P),\dotsi,A_{Q_0,p}(P)\right)=\left(\int_{Q_0}^p\omega_1,\dotsi,\int_{Q_0}^{P}\omega_p\right)\\
\end{split}
\end{equation}
and
$$\underline{\alpha}_{Q_0}: Div(\mathcal{K}_p)\rightarrow\mathcal{J}(\mathcal{K}_p),\mathcal{D}\mapsto\underline{\alpha}_{Q_0}(\mathcal{D})=\sum_{P\in\mathcal{K}_g}\mathcal{D}(P)\underline{A}_{Q_0}(P).$$

\begin{theorem}\label{TH4}
Suppose that $p, q$ satisfy the $n$-th stationary RLV hierarchy,let $P\in\mathcal{K}_{n}$\textbackslash$\{P_{\infty\pm},P_{0,\pm}\}$
and $(n,n_{0})\in\mathbb{Z}^{2}.$ Then $\mathcal{D}_{\underline{\hat{\mu}}(n)}$ is non-special. Moreover,\\
\begin{equation} \label{3.51}
\phi(P,n)=C(n)\frac{\theta(\underline{z}(P,\underline{\hat{\nu}}(n)))}{\theta(\underline{z}(P,\underline{\hat{\mu}}(n)))}\exp\left(\int_{Q_{0}}^{P}\omega_{P_{0,+}P_{\infty+}}^{(3)}\right),
\end{equation}
\begin{equation}\label{3.a0}
    \Psi_1(P,n,n_0)=C(n,n_0)\frac{\theta(\underline{z}(P,\underline{\hat{\mu}}(n)))}{\theta(\underline{z}(P,\underline{\hat{\mu}}(n_0)))}\\
    \times\exp\left(\left(n-n_0\right)\int_{Q_0}^{P}\omega_{P_{\sharp}(n) P_{\infty-}}\right),
\end{equation}
\begin{equation}\label{3.a1}
\begin{split}
    &\Psi_2(P,n,n_0)=\xi^{-1}\times
    C(n)C(n,n_0)\frac{\theta(\underline{z}(P,\underline{\hat{\nu}}(n)))}{\theta(\underline{z}(P,\underline{\hat{\mu}}(n_0)))}\\
    &\times\exp\left(\int_{Q_{0}}^{P}\omega_{P_{0,+}P_{\infty+}}^{(3)}+\left(n-n_0\right)\int_{Q_0}^{P}\omega_{P_{\sharp} (n)P_{\infty-}}\right),
\end{split}
\end{equation}
where
\begin{align*}\label{3.a2}
    C(n,n_0)=\frac{\theta(\underline{z}(P_{\infty-},\underline{\hat{\mu}}(n_0)))}{\theta(\underline{z}(P_{\infty-},\underline{\hat{\mu}}(n)))}
\end{align*}
and finally $p, q$ are the form of
\begin{equation}\label{3.100}
  p^{+}=\frac{1}{2}\left(-\Delta_3-\Delta_3^{+}-\delta_1+1-\Delta_1\pm\left((\Delta_3+\Delta_3^++\delta_1-1+\Delta_1)^{2}+4\Delta_2\right)^{\frac{1}{2}}\right)
\end{equation}
\begin{equation}\label{3.101}
  \begin{split}
  &q^{+}=\Delta_2/p^{+}+1\\
  &=2\Delta_2\left(-\Delta_3-\Delta_3^{+}-\delta_1+1-\Delta_1\pm\left((\Delta_3+\Delta_3^++\delta_1-1+\Delta_1)^{2}+4\Delta_2\right)^{-\frac{1}{2}}\right)^{-1}\\
  &+1.\\
  \end{split}
\end{equation}
Here
\begin{equation}\label{3.81}
  \Delta_1=\sum_{j=1}^n\lambda_j^{'}-\sum_{j=1}^n c_j(k)\partial_{\omega_j}\ln\left(\frac{\theta(\underline{z}(P_{\infty+},\underline{\hat{\mu}}(n))+\underline{\omega})}{\theta(\underline{z}(P_{\infty-},\underline{\hat{\mu}}(n))+\underline{\omega})}\right)|_{\underline{\omega}=0},
\end{equation}
\begin{equation}\label{3.82}
  \Delta_2=\kappa_{\infty+}-\sum_{j=1}^{n}c_{j}(n)\partial_{\omega_{j}}   \ln\left(\frac{\theta(\underline{z}(P_{\infty+},\underline{\hat{\nu}}(n))+\underline{\omega})}
   {\theta(\underline{z}(P_{\infty+},\underline{\hat{\mu}}(n))+\underline{\omega})}\right)
   |_{\underline{\omega}=0},
\end{equation}
\begin{equation}\label{3.83}
  \Delta_3=\frac{\theta(\underline{z}(P_{\infty-},\underline{\hat{\nu}}(n)))}{\theta(\underline{z}(P_{\infty-},\underline{\hat{\mu}}(n)))}\frac{\theta(\underline{z}(P_{\infty+},\underline{\hat{\mu}}(n)))}{\theta(\underline{z}(P_{\infty+},\underline{\hat{\nu}}(n)))}\frac{\tilde{a}_2}{\tilde{a}_1}
\end{equation}
and $\tilde{a}_1, \tilde{a}_2$, $ \{\lambda_j^{'}\}_{j=1,\dotsi,n}\in\mathbb{C}$ in (\ref{3.32}).

\end{theorem}

\begin{proof}
The proof that the divisor $\mathcal{D}_{\underline{\hat{\mu}}(n)}$ is non-special see Lemma \ref{lemma8}, where $t_r$ is regarded as a parameter. Hence the theta functions defined in this lemma are meaningful and not identical to zero.
Obviously, $$\tilde{\phi}(P,n)\frac{\theta(\underline{z}(P,\underline{\hat{\mu}}(n)))}{\theta(\underline{z}(P,\underline{\hat{\nu}}
(n)))}\exp\left(-\int_{Q_{0}}^{P}\omega_{P_{0,+}P_{\infty+}}^{(3)}\right)$$ is holomorphic function on compact Riemann surface $\mathcal{K}_{n}$(Riemann-Roch Theorem [14]). So it is a constant $C(n)$ related to $n$ and $\phi(P,n)$ has the form (\ref{3.51}).
One have the following expansion as $P\rightarrow P_{\infty+},\quad( z=\zeta^{-1})$

\begin{equation}
\begin{split}
&\frac{\theta(\underline{z}(P,\underline{\hat{\nu}}(n)))}{\theta(\underline{z}(P,\underline{\hat{\mu}}(n)))}=\frac{\theta(\underline{z}(P_{\infty+},\underline{\hat{\nu}}(n)))}{\theta(\underline{z}(P_{\infty+},\underline{\hat{\mu}}(n)))}\\
&\times\left(1-\sum_{j=1}^{n}c_{j}(n)\frac{\partial}{\partial\omega_{j}}\ln\left(
\frac{\theta(\underline{z}(P_{\infty+},\underline{\hat{\nu}}(n))+\underline{\omega})}{\theta(\underline{z}(P_{\infty+},\underline{\hat{\mu}}(n))+\underline{\omega})}\right)|_{\underline{\omega}=0}\zeta+O(\zeta^{2})\right).
\end{split}
\end{equation}
Then as $P\rightarrow P_{\infty+}$,
\begin{equation}\label{3.90}
\begin{split}
&\tilde{\phi}(P,n)\\
&=\tilde{a}_1C(n)\frac{\theta(\underline{z}(P_{\infty+},\underline{\hat{\nu}}(n)))}{\theta(\underline{z}(P_{\infty+},\underline{\hat{\mu}}(n)))}\\&\times\left(1-\sum_{j=1}^{n}c_{j}(n)\frac{\partial}{\partial\omega_{j}}\ln\left(\frac{\theta(\underline{z}(P_{\infty+},\underline{\hat{\nu}}(n))+\underline{\omega})}{\theta(\underline{z}(P_{\infty+},\underline{\hat{\mu}}(n))+\underline{\omega})}\right)|_{\underline{\omega}=0}\zeta+O(\zeta^{2})\right)\\
&\times\zeta^{-1}\left(1+\kappa_{\infty+}\zeta+O(\zeta^2)\right),
\end{split}
\end{equation}
where $\tilde{a}_1, \kappa_{\infty+}\in\mathbb{C}$ are constants generated in the limit procedure.
In another way, meromorphic function $\tilde{\phi}$ has the asymptotic expansion (\ref{3.28}).
Comparing the coefficients of (\ref{3.28}) with the ones of (\ref{3.90}) then yields
\begin{equation}\label{3.91}
  (\tilde{q}^{+})^{-1}=\tilde{a}_1C(n)\frac{\theta(\underline{z}(P_{\infty+},\underline{\hat{\nu}}(n)))}{\theta(\underline{z}(P_{\infty+},\underline{\hat{\mu}}(n)))}.
\end{equation}
and
\begin{equation}\label{3.93}
  \begin{split}
   &\left(\tilde{q}^{++}/\tilde{q}^{+}-1\right)\tilde{p}^{+}\\&=\kappa_{\infty+}-\sum_{j=1}^{n}c_{j}(n)\frac{\partial}{\partial\omega_{j}}
   \ln\left(\frac{\theta(\underline{z}(P_{\infty+},\underline{\hat{\nu}}(n))+\underline{\omega})}
   {\theta(\underline{z}(P_{\infty+},\underline{\hat{\mu}}(n))+\underline{\omega})}\right)
   |_{\underline{\omega}=0},\\
   \end{split}
\end{equation}
where $\underline{\omega}=(\omega_1,\omega_2,\dotsi,\omega_n).$
Similarly, as $P\rightarrow P_{\infty-}$, one finally derives
\begin{equation}\label{3.92}
  -\tilde{r}=\tilde{a}_2C(n)\frac{\theta(\underline{z}(P_{\infty-},\underline{\hat{\nu}}(n)))}{\theta(\underline{z}(P_{\infty-},\underline{\hat{\mu}}(n)))},
\end{equation}
where $\tilde{a}_2\in\mathbb{C}.$
Noted the assumption (\ref{2.1}) (\ref{3.1}) that $\tilde{q}(n)\neq 0, \tilde{r}(n)\neq 0$ for all $n\in\mathbb{Z}$, we conclude
\begin{equation}\label{3.94}
  -\tilde{q}^{+}\tilde{r}=\frac{\tilde{a}_2}{\tilde{a}_1}\frac{\theta(\underline{z}(P_{\infty-},\underline{\hat{\nu}}(n)))\theta(\underline{z}(P_{\infty+},\underline{\hat{\mu}}(n)))}{\theta(\underline{z}(P_{\infty-},\underline{\hat{\mu}}(n)))\theta(\underline{z}(P_{\infty+},\underline{\hat{\nu}}(n)))}
\end{equation}
from (\ref{3.91}) and (\ref{3.92}).
Let us consider the trace formula (\ref{3.26}). After a standard residue calculation at $P_{\infty\pm}$ \cite{A1,A2},
$\sum_{j=1}^{n}\mu_j$ has the following theta function representation
\begin{equation}\label{3.103}
  \sum_{j=1}^{n}\mu_j=\sum_{j=1}^n\lambda_j^{'}-\sum_{j=1}^n c_j(k)\partial_{\omega_j}\ln\left(\frac{\theta(\underline{z}(P_{\infty+},\underline{\hat{\mu}}(n))+\underline{\omega})}{\theta(\underline{z}(P_{\infty-},\underline{\hat{\mu}}(n))+\underline{\omega})}\right)|_{\underline{\omega}=0}
\end{equation}
Then $\tilde{p},\tilde{q} $ satisfy the following equations
\begin{equation}\label{3.104}
  \left(\tilde{q}^{++}/\tilde{q}^{+}-1\right)\tilde{p}^{+}=\Delta_2,
\end{equation}
\begin{equation}\label{3.79}
  -\tilde{q}^{+}\tilde{r}=\Delta_3,
\end{equation}
Plugging (\ref{3.104}) (\ref{3.79}) into (\ref{3.26}) then yields
\begin{equation}\label{3.80}
  (\tilde{p}^{+})^2+\left(\Delta_3+\Delta_3^++\Delta_1-\delta_1+1\right)\tilde{p}^{+}-\Delta_2=0.
\end{equation}
The two solutions of (\ref{3.80}) are (\ref{3.100}). (\ref{3.101}) follows from (\ref{3.104}) and (\ref{3.100}).

\end{proof}
\begin{remark}\label{remark5}
The theta function representations of $\Psi_1, \Psi_2$ in (\ref{3.a0}) (\ref{3.a1}) is merely a formal expression
and they are of little use in the process of solving the stationary relativistic Lotka-Volterra system (\ref{2.34}).
In fact the expression (\ref{3.a0}) itself is indeterminate as $P_{\sharp}$ is unknown and may relate to the lattice
variant $n$.
\end{remark}

\section{Algebro-geometric Solutions of Time-dependent Relativistic Lotka-Volterra Hierarchy}
In this section, we mainly derive the algebro-geometric solutions of time-dependent relativistic Lotka-Volterra hierarchy defined in section 2 by extending the method employed in the section 3 to time-dependent cases.
\begin{guess}
\text{Assume}
\begin{equation}\label{4.1}
    \begin{split}
    &p=\tilde{p}\neq 0,1\quad p(\cdot,t),q(\cdot,t),\tilde{p}(\cdot,t),\tilde{q}(\cdot,t)\tilde{r}(\cdot,t)\in\ell(\mathbb{Z}),\\
    &p(n,\cdot),q(n,\cdot),\tilde{p}(n,\cdot),\tilde{q}(n,\cdot),\tilde{r}(n,\cdot)\in\textbf{C}^{1}(\mathbb{R})\\
    &\text{and}\quad (n,t)\in\mathbb{Z}\times\mathbb{R}.\\
    \end{split}
\end{equation}
\end{guess}
Throughout this section we suppose Hypothesis 2 and Hypothesis 3 holds.
The basic algebro-geometric initial value problem is that if we consider a solution $p^{1}(n), q^{1}(n)$ of the $n$-th stationary relativistic
Lotka-Volterra system $s$-$RLV_p(p^{1},q^{1})$\\=0, associated with the hyperelliptic curve $\mathcal{K}_n$ and a corresponding of the summation $\{\delta_\ell\}_{\ell=0}^{n}\subseteq\mathbb{C}$, then we construct a solution $p,q$ of the $r$-th time-dependent relativistic Lotka-Volterra
flow $RLV_r(p,q)=0$ satisfying $p(n,t_{0,r})=p^{1}(n), q(n,t_{0,r})$\\$=q^{1}(n)$ for some $t_{0,r}\in\mathbb{R}$ and any $n\in\mathbb{Z}.$
We shall use the notation $\bar{\tilde{V}}_r,\bar{\tilde{A}}_{r+1},\bar{\tilde{B}}_{r},$ $\bar{\tilde{C}}_{r},\bar{\tilde{D}}_{r+1},\bar{\tilde{a}}_{\ell},\bar{\tilde{b}}_{\ell},
\bar{\tilde{c}}_{\ell},\bar{\tilde{d}}_{\ell},\bar{\tilde{\delta}}_{\ell}$ in the $r$-th time-dependent flow to distinguish $\tilde{A}_{n+1},\tilde{B}_{n},$\\ $\tilde{C}_{n},\tilde{D}_{n+1},
\tilde{a}_{\ell},\tilde{b}_{\ell},\tilde{c}_{\ell},\tilde{d}_{\ell},\tilde{\delta}_{\ell}$ in the $n$-th stationary relativistic Lotka-Volterra system.
The algebro-geometric initial value problem discussed above can be summed up in the form of zero-curvature equation
\begin{equation}\label{4.2}
     \tilde{U}_{t_{r}}(\xi,t_{r})+\tilde{U}(\xi,t_{r})\bar{\tilde{V}}_{r}(\xi,t_{r})-\bar{\tilde{V}}^{+}_{r}(\xi,t_{r})U(\xi,t_{r})=0,
\end{equation}
\begin{equation}\label{4.3}
     \tilde{U}(\xi,t_{0,r})\tilde{V}_{n}(\xi,t_{0,r})-\tilde{V}^{+}_{n}(\xi,t_{0,r})\tilde{U}(\xi,t_{0,r})=0,
\end{equation}
Considering the isospectral property of the Lax operator $L$ corresponding to $U$, we may impose more strong condition
on equation (\ref{4.3}), which is
\begin{equation}\label{4.4}
     \tilde{U}(\xi,t_{r})\tilde{V}_{n}(\xi,t_{r})-\tilde{V}^{+}_{n}(\xi,t_{r})\tilde{U}(\xi,t_{r})=0, \quad t_r\in\mathbb{R}.
\end{equation}
For further reference, we recall the relevant quantities here:
\begin{equation}\label{4.5}
\tilde{U}(\xi)=\left(
    \begin{array}{cc}
       \tilde{p}-z & \tilde{q}\xi\\
        \tilde{r}\xi & 1
   \end{array}
   \right),
\end{equation}
\begin{equation}\label{4.6}
\tilde{V}_{n}(\xi)=\left(
    \begin{array}{cc}
       \tilde{A}_{n+1}^-(z)& \xi\tilde{B}_{n}^-(z)\\
        \xi\tilde{C}_{n}^-(z) & -\tilde{D}_{n+1}^-(z)
   \end{array}
   \right),
\end{equation}
\begin{equation}\label{4.7}
\tilde{V}_{r}(\xi)=\left(
    \begin{array}{cc}
       \bar{\tilde{A}}_{r+1}^-(z)& \xi\bar{\tilde{B}}_{r}^-(z)\\
        \xi\bar{\tilde{C}}_{r}^-(z) &-\bar{\tilde{D}}_{r+1}^-(z)
   \end{array}
   \right),
\end{equation}
where
\begin{equation}\label{4.8}
    \tilde{A}_{n+1}(z)=\sum_{\ell=0}^{n+1}\tilde{a}_{n+1-\ell}z^{\ell},
\end{equation}
\begin{equation}\label{4.9}
    \tilde{B}_{n}(z)=\sum_{\ell=0}^{n}\tilde{b}_{n-\ell}z^{\ell},
\end{equation}
\begin{equation}\label{4.10}
     \tilde{C}_{n}(z)=\sum_{\ell=0}^{n}\tilde{c}_{n-\ell}z^{\ell},
\end{equation}
\begin{equation}\label{4.11}
    \tilde{D}_{n+1}(z)=\sum_{\ell=0}^{n+1}\tilde{a}_{n+1-\ell}z^{\ell},
\end{equation}
\begin{equation}\label{4.12}
    \bar{\tilde{A}}_{r+1}(z)=\sum_{\ell=0}^{r+1}\bar{\tilde{a}}_{r+1-\ell}z^{\ell},
\end{equation}
\begin{equation}\label{4.13}
    \bar{\tilde{B}}_{r}(z)=\sum_{\ell=0}^{r}\bar{\tilde{b}}_{r-\ell}z^{\ell},
\end{equation}
\begin{equation}\label{4.14}
     \bar{\tilde{C}}_{r}(z)=\sum_{\ell=0}^{r}\bar{\tilde{c}}_{r-\ell}z^{\ell},
\end{equation}
\begin{equation}\label{4.15}
    \bar{\tilde{D}}_{r+1}(z)=\sum_{\ell=0}^{r+1}\bar{\tilde{a}}_{r+1-\ell}z^{\ell},
\end{equation}
Here $\{\tilde{a}_\ell\}_{\ell=0}^{n+1}, \{\tilde{b}_\ell\}_{\ell=0}^{n}, \{\tilde{c}_{\ell}\}_{\ell=0}^{n}, \{\tilde{d}_{\ell}\}_{\ell=0}^{n+1}$ and
$\{\bar{\tilde{a}}_\ell\}_{\ell=0}^{n+1}, \{\bar{\tilde{b}}_\ell\}_{\ell=0}^{n}, \{\bar{\tilde{c}}_{\ell}\}_{\ell=0}^{n}, \{\bar{\tilde{d}}_{\ell}\}_{\ell=0}^{n+1}$ are defined by (\ref{2.2})-(\ref{2.6}) corresponding to different constants $\tilde{\delta}_\ell$ and $\bar{\tilde{\delta}}_\ell$, respectively.
Explicitly, equation (\ref{4.2}) and (\ref{4.4}) are equivalent to
\begin{equation}\label{4.16}
     -\tilde{p}_{t_r}=(\tilde{p}-z)\bar{\tilde{A}}_{r+1}^-+z\tilde{q}\bar{\tilde{C}}_r^{-}-
     (\tilde{p}-z)\bar{\tilde{A}}_{r+1}-z\tilde{r}\bar{\tilde{B}}_{r},
   \end{equation}
   \begin{equation}\label{4.17}
     -\tilde{q}_{t_r}=(\tilde{p}-z)\bar{\tilde{B}}_{r}^--\tilde{q}\bar{\tilde{D}}_{r+1}^--\tilde{q}\bar{\tilde{A}}_{r+1}-\bar{\tilde{B}}_{r},
   \end{equation}
   \begin{equation}\label{4.18}
     -\tilde{r}_{t_r}=\tilde{r}\bar{\tilde{A}}_{r+1}^-+\bar{\tilde{C}}_{r}^--(\tilde{p}-z)\bar{\tilde{C}}_{r}+\tilde{r}\bar{\tilde{D}}_{r+1},
   \end{equation}
   \begin{equation}\label{4.19}
    0= z\tilde{r}\bar{\tilde{B}}_{r}^--\bar{\tilde{D}}_{r+1}^--z\tilde{q}\bar{\tilde{C}}_{r}+\bar{\tilde{D}}_{r+1},
   \end{equation}
   \begin{equation}\label{4.20}
     0=(\tilde{p}-z)\tilde{A}_{n+1}^-+z\tilde{q}\tilde{C}_n^{-}-
     (\tilde{p}-z)\tilde{A}_{n+1}-z\tilde{r}\tilde{B}_{n+1},
   \end{equation}
   \begin{equation}\label{4.21}
     0=(\tilde{p}-z)\tilde{B}_{n}^--\tilde{q}\tilde{D}_{n+1}^--\tilde{q}\tilde{A}_{n+1}-\tilde{B}_{n},
   \end{equation}
   \begin{equation}\label{4.22}
     0=\tilde{r}\tilde{A}_{n+1}^-+\tilde{C}_{n}^--(\tilde{p}-z)\tilde{C}_{n}+\tilde{r}\tilde{D}_{n+1},
   \end{equation}
   \begin{equation}\label{4.23}
     0=z\tilde{r}\tilde{B}_{n}^--\tilde{D}_{n+1}^--z\tilde{q}\tilde{C}_{n}+\tilde{D}_{n+1},
   \end{equation}
   respectively.
   In particular, (\ref{2.62}) holds in the present $t_r$-dependence setting, that is,
    \begin{equation}\label{4.24}
         R_{2n+2}(z,t_r)=-(\tilde{A}_{n+1}(z,t_r))^2-z\tilde{B}_n(z,t_r)\tilde{C}_n(z,t_r).
    \end{equation}
    Obviously the algebraic curve defined in (\ref{4.24}) is $n$-independent and may depend on the parameter $t_{r}.$
    In fact we can prove $$\partial_{t_r}R_{2n+2}(z,t_r)=0$$ under the initial value condition (\ref{4.4}), which means
    $R_{2n+2}(z,t_r)$ is $t_r$-independent (see lemma \ref{lemma5}).
    We write
    \begin{equation}\label{4.24a}
    \begin{split}
    &\tilde{B}_{n}(z,n,t_r)=\left(-\tilde{q}^{+}(n,t_r)\right)\prod_{j=1}^{n}\left(z-\mu_j(n,t_r)\right),\\ &\tilde{C}_{n}(z,n,t_r)=\left(-\tilde{r}(n,t_r)\right)\prod_{j=1}^{n}\left(z-\nu_j(n,t_r)\right),\\
    &\text{and}\quad\mu_j(n,t_r),\nu_j(n,t_r)\in\ell(\mathbb{Z}), (n,t_r)\in\mathbb{Z}\times\mathbb{R}.
    \end{split}
    \end{equation}
    As in the stationary context (\ref{3.5}) (\ref{3.7}) we introduce
    \begin{equation}\label{4.25}
    \begin{split}
   &\hat{\mu}_j(n,t_r)=(\mu_j(n,t_r),-2\tilde{A}_{n+1}(\mu_j(n,t_r),n,t_r)),\\ &\hat{\nu}_j(n,t_r)=(\nu_j(n,t_r),2\tilde{A}_{n+1}(\nu_j(n,t_r),n,t_r)),\\
   &j=1,2,\dotsi,n, (n,t_r)\in\mathbb{Z}\times\mathbb{R}.
   \end{split}
   \end{equation}
   on $\mathcal{K}_n$ and define the following meromorphic function $\tilde{\phi}(\cdot,n,t_r)$ on $\mathcal{K}_n$,
   \begin{equation}\label{4.26}
    \tilde{\phi}(P,n,t_r)=\frac{\frac{1}{2}y-\tilde{A}_{n+1}(z,n,t_r)}{B_{n}(z,n,t_r)}
    =\frac{z\tilde{C}_{n}(z,n,t_r)}{\frac{1}{2}y+\tilde{A}_{n+1}(z,n,t_r)},
    \quad P=(z,y)\in\mathcal{K}_n.
\end{equation}
The divisor $(\tilde{\phi}(\cdot))$ of $\tilde{\phi}$ is
\begin{equation}\label{4.26a}
    (\tilde{\phi}(\cdot,n))=\mathcal{D}_{P_{0,+}\underline{\hat{\nu}}(n,t_r)}-\mathcal{D}_{P_{\infty+}\underline{\hat{\mu}}(n,t_r)},
\end{equation}
and the time-dependent Baker-Akhiezer vector is then defined in term of $\tilde{\phi}$ by
\begin{equation}\label{4.27}
 \Psi(P,\xi,n,n_0,t_r,t_{0,r})=\left(
                                  \begin{array}{c}
                                     \Psi_1(P,\xi,n,n_0,t_r,t_{0,r}) \\
                                     \Psi_2(P,\xi,n,n_0,t_r,t_{0,r}) \\
                                  \end{array}
                                \right),
\end{equation}
\begin{equation}\label{4.28}
\begin{split}
 &\Psi_1(P,\xi,n,n_0,t_r,t_{0,r})
 =\exp\left(\int_{t_{0,r}}^{t_r}\left(\bar{\tilde{A}}_{r+1}(z,n_0,s)+\bar{\tilde{B}}_{r}(z,n_0,s)\tilde{\phi}(P,n_0,s)\right)ds
 \right)\\
 &\times\begin{cases}
      \prod_{n^{'}=n_0+1}^{n}\left(\tilde{p}(n^{'},t_r)-z+\tilde{q}(n^{'})\tilde{\phi}^-(P,n^{'},t_r)\right),&n\geq n_0+1,\cr
      1,&n=n_0,\cr
      \prod_{n^{'}=n+1}^{n_0}\left(\tilde{p}(n^{'},t_r)-z+\tilde{q}(n^{'},t_r)\tilde{\phi}^-(P,n^{'},t_r)\right)^{-1},&n\leq n_0-1,
    \end{cases}\\
    \end{split}
\end{equation}
\begin{eqnarray}\label{4.29}
\begin{split}
 &\Psi_2(P,\xi,n,n_0,t_r,t_{0,r})=\exp\left(\int_{t_{0,r}}^{t_r}\left(\bar{\tilde{A}}_{r+1}(z,n_0,s)+\bar{\tilde{B}}_{r}(z,n_0,s)
 \tilde{\phi}(P,n_0,s)\right)ds\right)\\
 &\times\xi^{-1}\times\tilde{\phi}(P,n_0,t_r)
    \begin{cases}
      \prod_{n^{'}=n_0+1}^{n}\left(\frac{\tilde{r}(n^{'},t_r)z}{\tilde{\phi}^-(P,n^{'},t_r)}+1\right),&n\geq n_0+1,\cr
      1,&n=n_0,\cr
      \prod_{n^{'}=n+1}^{n_0}\left(\frac{\tilde{r}(n^{'},t_r)z}{\tilde{\phi}^-(P,n^{'},t_r)}+1\right)^{-1},&n\leq n_0-1.
    \end{cases}\\
 \end{split}
 \end{eqnarray}
 \begin{equation}\label{4.30}
   P=(z,y)\in\mathcal{K}_n\backslash\{P_{0,\pm},P_{\infty\pm}\},\quad (n,t_r)\in\mathbb{Z}\times\mathbb{R}.
 \end{equation}
 One observes that
\begin{equation}\label{4.31}
\begin{split}
  &\Psi_1(P,n,n_0,t_r,t_{0,r})=\Psi_1(P,n_0,n_0,t_r,t_{0,r})\Psi_1(P,n,n_0,t_r,t_{0,r}),\\
  &P=(z,y)\in\mathcal{K}_n\backslash\{P_{0,\pm},P_{\infty\pm}\},\quad (n,n_0,t_r,t_{0,r})\in\mathbb{Z}^2\times\mathbb{R}^2.\\
\end{split}
\end{equation}
The following lemma shows the properties of $\tilde{\phi}, \Psi_1$ and $\Psi_2$ as discussed in the stationary case.\
\begin{lemma}
Assume Hypothesis 2 and Hypothesis 3 hold and
suppose $p(n,t_r),q(n,t_r)$ satisfy (\ref{4.2}) (\ref{4.4}). In addition,  let $P=(z,y)\in\mathcal{K}_n\backslash\{P_{\infty\pm},P_{0,\pm}\}, (n,n_0,t_r,t_{0,r})\in\mathbb{Z}^2\times\mathbb{R}^2.$ Then $\tilde{\phi}(P,t_r)$ satisfies the following equations
\begin{equation}\label{4.32}
    (\tilde{p}-z)\tilde{\phi}(P)+\tilde{q}\tilde{\phi}(P)\tilde{\phi}^-(P)=\tilde{r}z+\tilde{\phi}^{-}(P),
\end{equation}
\begin{equation}\label{4.33}
  \tilde{\phi}_{t_r}(P)=z\bar{\tilde{C}}_r-\bar{\tilde{D}}_{r+1}\tilde{\phi}(P)-\bar{\tilde{A}}_{r+1}\tilde{\phi}(P)-\bar{\tilde{B}}_r\tilde{\phi}^2(P),
\end{equation}
as well as
    \begin{equation}\label{4.34}
        \tilde{\phi}(P)\tilde{\phi}(P^{*})=\frac{-z\tilde{C}_{n}(z)}{\tilde{B}_{n}(z)},
    \end{equation}
    \begin{equation}\label{4.35}
        \tilde{\phi}(P)+\tilde{\phi}(P^{*})=\frac{-2\tilde{A}_{n+1}}{\tilde{B}_n},
   \end{equation}
   \begin{equation}\label{4.36}
        \tilde{\phi}(P)-\tilde{\phi}(P^{*})=\frac{y}{\tilde{B}_n(z)}.
   \end{equation}
   The vector $\Psi$ satisfy
   \begin{equation}\label{4.37}
     \tilde{U}(\xi)\Psi^-(P)=\Psi(P),
   \end{equation}
   \begin{equation}\label{4.38}
     \tilde{V}_n(\xi)\Psi(P)=(1/2)y\Psi(P),
   \end{equation}
    \begin{equation}\label{4.39}
        \Psi_2(P,\xi,n,n_0)=\xi^{-1}\tilde{\phi}(P,n)\Psi_1(P,\xi,n,n_0),
    \end{equation}
    \begin{equation}\label{4.40}
       \Psi_{t_r}(P)=\bar{\tilde{V}}_r(\xi)\Psi(P).
    \end{equation}
    Moreover, Moreover, as long as the zeros of $\mu_j(n_0,s)$ of $B_n(\cdot,n_0,s)$ are all simple and nonzero for all $s\in\Omega, \Omega\subseteq\mathbb{R}$ is an open interval, $\Psi_1$ is meromorphic on $\mathcal{K}_n\backslash\{P_{0,\pm}, P_{\infty\pm}\}$ for $(n,t_r,t_{0,r})\in\mathbb{Z}\times\Omega^2.$

\end{lemma}
\begin{proof}
The proof of (\ref{4.32}), (\ref{4.34})-(\ref{4.39}) is the same with lemma \ref{lemma1}, where $t_r$
is regarded as a parameter.
From (\ref{4.32}) we have
\begin{equation}\label{4.41}
  \tilde{p}_{t_r}\tilde{\phi}+(\tilde{p}-z)\tilde{\phi}+\tilde{q}_{t_r}\tilde{\phi}\tilde{\phi}^-
  +\tilde{q}\tilde{\phi}_{t_r}\tilde{\phi}^{-}+\tilde{q}\tilde{\phi}\tilde{\phi}^{-}_{t_r}=\tilde{r}_{t_r}z
  +\tilde{\phi}^{-}_{t_r},
\end{equation}
that is,
\begin{equation}\label{4.42}
  [\tilde{p}-z+\tilde{q}\tilde{\phi}^{-}+(\tilde{q}\tilde{\phi}-1)S^{-}]\tilde{\phi}_{t_r}=\tilde{r}_{t_r}z-
  \tilde{p}_{t_r}\tilde{\phi}-
  \tilde{q}_{t_r}\tilde{\phi}\tilde{\phi}^{-}.
\end{equation}
Using (\ref{4.16})-(\ref{4.18}), one finds
\begin{equation}\label{4.43}
  \begin{split}
  &[\tilde{p}-z+\tilde{q}\tilde{\phi}^-+(\tilde{q}\tilde{\phi}-1)S^{-}]\left(z\bar{\tilde{C}}_r-\bar{\tilde{D}}_{r+1}
  \tilde{\phi}-\bar{\tilde{A}}_{r+1}\tilde{\phi}-\bar{\tilde{B}}_r\tilde{\phi}^2\right)\\
  &=[\frac{\tilde{r}z+\tilde{\phi}^{-}}{\tilde{\phi}}+(\tilde{q}\tilde{\phi}-1)S^{-}]
  \left(z\bar{\tilde{C}}_r-\bar{\tilde{D}}_{r+1}\tilde{\phi}-
  \bar{\tilde{A}}_{r+1}\tilde{\phi}-\bar{\tilde{B}}_r\tilde{\phi}^2\right)\\
  &=(\tilde{p}-z+\tilde{q}\tilde{\phi}^{-})z\bar{\tilde{C}}_{r}-(\tilde{r}z+\tilde{\phi}^{-})\bar{\tilde{D}}_{r+1}
  -(\tilde{r}z+\tilde{\phi}^{-})\bar{\tilde{A}}_{r+1}-(\tilde{r}z+\tilde{\phi}^{-})\bar{\tilde{B}}_r\tilde{\phi}^-\\
  &+(\tilde{q}\tilde{\phi}-1)z\bar{\tilde{C}}_{r}^--(\tilde{q}\tilde{\phi}-1)\bar{\tilde{D}}_{r+1}^-\tilde{\phi}^{-}
  -(\tilde{q}\tilde{\phi}-1)\bar{\tilde{A}}_{r+1}^-\tilde{\phi}^{-}-(\tilde{q}\tilde{\phi}-1)\bar{\tilde{B}}_r
  \tilde{\phi}^2\\
  &=(\tilde{p}-z+\tilde{q}\tilde{\phi}^{-})z\bar{\tilde{C}}_{r}-(\tilde{r}z+\tilde{\phi}^{-})\bar{\tilde{D}}_{r+1}
  -(\tilde{r}z+\tilde{\phi}^{-})\bar{\tilde{A}}_{r+1}-(\tilde{r}z+\tilde{\phi}^{-})\bar{\tilde{B}}_r\tilde{\phi}^-\\
  &+(\tilde{q}\tilde{\phi}-1)z\bar{\tilde{C}}_{r}^--(\tilde{q}\tilde{\phi}-1)\bar{\tilde{D}}_{r+1}^-\tilde{\phi}^{-}
  -(\tilde{q}\tilde{\phi}-1)\bar{\tilde{A}}_{r+1}^-\tilde{\phi}^{-}-\tilde{r}z\bar{\tilde{B}}_r^-\tilde{\phi}^-
  +\bar{\tilde{B}}_r(\tilde{p}-z)\tilde{\phi}\tilde{\phi}^-\\
  &=[\tilde{r}\bar{\tilde{A}}_{r+1}^--\bar{\tilde{C}}_{r}^-+(\tilde{p}-z)\bar{\tilde{C}}_{r}-\tilde{r}\bar{\tilde{D}}_{r+1}]z+
  [(\tilde{p}-z)\bar{\tilde{A}}_{r+1}^-+z\tilde{q}\bar{\tilde{C}}_r^{-}
  -(\tilde{p}-z)\bar{\tilde{A}}_{r+1}\\&-z\tilde{r}\bar{\tilde{B}}_{r+1}]\tilde{\phi}+[(\tilde{p}-z)\bar{\tilde{B}}_{r}^-
  -\tilde{q}\bar{\tilde{D}}_{r+1}^--\tilde{q}\bar{\tilde{A}}_{r+1}-\bar{\tilde{B}}_{r}]\tilde{\phi}\tilde{\phi}^-\\
  &+[(\bar{\tilde{A}}_{r+1}^--\bar{\tilde{A}}_{r+1})(\tilde{\phi}^--(\tilde{p}-z)\tilde{\phi}-\tilde{q}\tilde{\phi}\tilde{\phi}^-+\tilde{r}z)]
  +[z\tilde{r}\bar{\tilde{B}}_{r}^--\bar{\tilde{D}}_{r+1}^--z\tilde{q}\bar{\tilde{C}}_{r}+\bar{\tilde{D}}_{r+1}]\tilde{\phi}^-\\
  &=\tilde{r}_{t_r}z-\tilde{p}_{t_r}\tilde{\phi}-\tilde{q}_{t_r}\tilde{\phi}\tilde{\phi}^-.\\
  \end{split}
\end{equation}
Considering (\ref{4.42}), one may conclude
\begin{equation}\label{4.44}\begin{split}
    &[\tilde{p}-z+\tilde{q}\tilde{\phi}^{-}+(\tilde{q}\tilde{\phi}-1)S^{-}]\left(\tilde{\phi}_{t_r}-(z\bar{\tilde{C}}_r-\bar{\tilde{D}}_{r+1}\tilde{\phi}-\bar{\tilde{A}}_{r+1}\tilde{\phi}-\bar{\tilde{B}}_r\tilde{\phi}^2)\right) \\&=0.\\\end{split}
\end{equation}
Thus, the expression
\begin{equation}\label{4.45}
 \tilde{\phi}_{t_r}-(z\bar{\tilde{C}}_r-\bar{\tilde{D}}_{r+1}\tilde{\phi}-\bar{\tilde{A}}_{r+1}\tilde{\phi}
  -\bar{\tilde{B}}_r\tilde{\phi}^2)\\
  =C\begin{cases}\prod_{n^{'}=n_0+1}^{n}B(z,n^{'},t_r),&n\geq n_0+1,\cr
  1,&n=n_0,\cr
  \prod_{n^{'}=n_0}^{n}B(z,n^{'},t_r)^{-1},&n\leq n_0-1,
  \end{cases}
 \end{equation}
 where
 \begin{equation}\label{4.46}
   B(z,n,t_r)=\frac{1-\tilde{q}(n,t_r)\tilde{\phi}(z,n,t_r)}{\tilde{p}(n,t_r)-z+\tilde{q}(n,t_r)
   \tilde{\phi}^{-}(z,n,t_r)},\quad (n,t_r)\in\mathbb{Z}\times\mathbb{R}.
 \end{equation}
 Obviously, the expression $\tilde{\phi}_{t_r}-(z\bar{\tilde{C}}_r-\bar{\tilde{D}}_{r+1}\tilde{\phi}-\bar{\tilde{A}}_{r+1}\tilde{\phi}
  -\bar{\tilde{B}}_r\tilde{\phi}^2)$ in the left-hand of (\ref{4.45}) is a meromorphic function on Riemann Surface $\mathcal{K}_n$ and its order is finite at $P_{\infty-}$. However, the order of the right-hand in (\ref{4.45}) is $O(z^{n-n_0})$ as $P$ near the point $P_{\infty-}.$ Hence, taking $n$ sufficiently large, then yields a contradiction at the both bides of (\ref{4.45}) unless $C=0.$ Therefore we have (\ref{4.33}). From the definition of
  $\Psi_1(P,\xi,n,n_0,t_r,t_{0,r})$ in (\ref{4.28}) it is easy to check
  \begin{equation}\label{4.47}
   \begin{split}
    &(\tilde{p}-z+\tilde{q}\tilde{\phi}^-)_{t_r}=\tilde{p}_{t_r}+\tilde{q}_{t_r}\tilde{\phi}^-
    +\tilde{q}\tilde{\phi}^-_{t_r}\\
    &=-(\tilde{p}-z)\bar{\tilde{A}}_{r+1}^--z\tilde{q}\bar{\tilde{C}}_r^{-}+
     (\tilde{p}-z)\bar{\tilde{A}}_{r+1}+z\tilde{r}\bar{\tilde{B}}_{r}\\
    &\times\left(-(\tilde{p}-z)\bar{\tilde{B}}_{r}^-+\tilde{q}\bar{\tilde{D}}_{r+1}^-+\tilde{q}\bar{\tilde{A}}_{r+1}+
    \bar{\tilde{B}}_{r}\right)\tilde{\phi}^-\\
    &+\tilde{q}\left(z\bar{\tilde{C}}_r^--\bar{\tilde{D}}_{r+1}^-
    \tilde{\phi}^--\bar{\tilde{A}}_{r+1}^-\tilde{\phi}^--\bar{\tilde{B}}_r^-(\tilde{\phi}^-)^2\right)\\
    &=\left(\tilde{p}-z+\tilde{q}\tilde{\phi}^-\right)\left(\bar{\tilde{A}}_{r+1}+\bar{\tilde{B}}_r\tilde{\phi}-
    \bar{\tilde{A}}_{r+1}^--\bar{\tilde{B}}_r^-\tilde{\phi}^-\right).\\
    \end{split}
  \end{equation}\
  If we note
  \begin{equation}\label{4.48}
    \Upsilon(P,n_0,t_r,t_{0,r})=\exp\left(\int_{t_{0,r}}^{t_r}\left(\bar{\tilde{A}}_{r+1}(z,n_0,s)+\bar{\tilde{B}}_{r}(z,n_0,s)\tilde{\phi}(P,n_0,s)\right)ds
 \right),
  \end{equation}
  then
  \begin{equation}\label{4.49}
  \begin{split}
    &\Psi_{1,t_r}=\left(\Upsilon\prod_{n^{'}=n_0+1}^{n}(\tilde{p}-z+\tilde{q}\tilde{\phi}^-)(n^{'})\right)_{t_r}\\
    &=\Upsilon_{t_r}\prod_{n^{'}=n_0+1}^{n}(\tilde{p}-z+\tilde{q}\tilde{\phi}^-)(n^{'})+\Upsilon\sum_{n^{'}=n_0+1}^{n}
    (\tilde{p}-z+\tilde{q}\tilde{\phi}^-)_{t_r}(n^{'})\prod_{n^{''}\neq n^{'}}(\tilde{p}-z+\tilde{q}\tilde{\phi}^-)(n^{''})\\
    &=\Upsilon\prod_{n^{'}=n_0+1}^{n}(\tilde{p}-z+\tilde{q}\tilde{\phi}^-)(n^{'})[\bar{\tilde{A}}_{r+1}(n_0,t_r)+
    \sum_{n^{'}=n_0+1}^{n}(\bar{\tilde{A}}_{r+1}(n^{'})+\bar{\tilde{B}}_r(n^{'})-\\
    &\bar{\tilde{A}}_{r+1}^-(n^{'})-\bar{\tilde{B}}_r^-(n^{'}))]\\
    &=(\bar{\tilde{A}}_{r+1}+\bar{\tilde{B}}_{r})\Psi_1=\bar{\tilde{A}}_{r+1}\Psi_1+\xi\bar{\tilde{B}}_{r}\Psi_2,\\
    &\text{as} \quad n\geq n_0+1.\\
  \end{split}
  \end{equation}
  The proof for the case $n\leq n_0-1$ is similarly with (\ref{4.49}).
  Using (\ref{4.32}) and (\ref{4.49}) we have
  \begin{equation}\label{4.50}
  \begin{split}
    &\Psi_{2,t_r}=\xi^{-1}(\tilde{\phi}_{t_r}\Psi_1+\tilde{\phi}\Psi_{1,t_r})\\
    &=\xi^{-1}\left(z\bar{\tilde{C}}_r-\bar{\tilde{D}}_{r+1}\tilde{\phi}-
  \bar{\tilde{A}}_{r+1}\tilde{\phi}-\bar{\tilde{B}}_r\tilde{\phi}^2+(\bar{\tilde{A}}_{r+1}+\bar{\tilde{B}}_{r+1}
  \tilde{\phi})\tilde{\phi}\right)\Psi_1\\
  &=\xi^{-1}(z\bar{\tilde{C}}_r-\bar{\tilde{D}}_{r+1})\Psi_1=\xi\bar{\tilde{C}}_r\Psi_1-\bar{\tilde{D}}_{r+1}\Psi_2.\\
  \end{split}
  \end{equation}
  To illustrate that $\Psi_1(\cdot,\xi,n,n_0,t_r,t_{0,r})$ is a meromorphic function on $\mathcal{K}_n\backslash\{P_{0,\pm},P_{\infty\pm}\}$ we only need to prove $\Upsilon$ is a meromorphic function on
  $\mathcal{K}_n\backslash\{P_{0,\pm},P_{\infty\pm}\}$. Taking into account (\ref{4.25}) (\ref{4.26}) (\ref{4.52}), one may derive
  \begin{equation}\label{4.51}
    \bar{\tilde{B}}_r(z,n_0,s)\tilde{\phi}(P,n_0,s)\thicksim \partial_s\ln\left(B_r(z,n_0,s)\right)+O(1)
  \end{equation}
  as $P\rightarrow\mu_j(n_0,s).$ Hence we conclude $\Upsilon$ is a meromorphic function on $\mathcal{K}_n\backslash\{P_{0,\pm},P_{\infty\pm}\}$ with the help of analysis the asymptotic
  behavior of the possible poles $\{\mu_j(n_0,s)\}$ of the function $\bar{\tilde{A}}_{r+1}(z,n_0,s)+\bar{\tilde{B}}_{r}(z,n_0,s)\tilde{\phi}(P,n_0,s).$

  Next we consider the $t_r$-dependence of $\tilde{A}_{n+1}, \tilde{B}_n, \tilde{C}_n.$
  \begin{lemma}\label{lemma5}
  Assume Hypothesis 2 and Hypothesis 3 hold and
suppose $p(n,t_r),q(n,t_r)$ satisfy (\ref{4.2}) (\ref{4.4}). In addition, let $(z,n,t_r)\in\mathbb{C}\times\mathbb{Z}\times\mathbb{R}.$
 Then,
 \begin{equation}\label{4.52}
    \tilde{\bar{B}}_{n,t_r}=-2\tilde{A}_{n+1}\bar{\tilde{B}}_r+(\bar{\tilde{A}}_{r+1}+\bar{\tilde{D}}_{r+1})\tilde{B}_n,
  \end{equation}
  \begin{equation}\label{4.53}
    \tilde{A}_{n+1,t_r}=z(\bar{\tilde{B}}_r\tilde{C}_n-\tilde{B}_n\bar{\tilde{C}}_r),
  \end{equation}
  \begin{equation}\label{4.54}
    \tilde{C}_{n,t_r}=2z^{-1}\tilde{A}_{n+1}\bar{\tilde{C}}_r-\tilde{C}_n(\bar{\tilde{A}}_{r+1}+\bar{\tilde{D}}_{r+1}).
  \end{equation}
  In particular, (\ref{4.52})-(\ref{4.54}) is equivalent to
  \begin{equation}\label{4.55}
    \tilde{V}_{n,t_r}=[\bar{\tilde{V}}_r,\tilde{V}_n].
  \end{equation}
  and the algebraic curve defined in (\ref{4.24}) is $t_r$-independent.
  \end{lemma}
  \begin{proof}
  To prove (\ref{4.52}) one first differentiates equation (\ref{4.36})
  \begin{equation}\label{4.56}
    \tilde{\phi}_{t_r}(P)-\tilde{\phi}_{t_r}(P^{*})=\frac{-y\tilde{B}_{n,t_r}}{\tilde{B}_n^2}.
  \end{equation}
  The time derivative of $\tilde{\phi}$ given in (\ref{4.32}) and (\ref{4.35}) (\ref{4.36}) yields
  \begin{equation}\label{4.57}
  \begin{split}
    &\tilde{\phi}_{t_r}(P)-\tilde{\phi}_{t_r}(P^{*})\\
    &=\bar{\tilde{B}}_r(\tilde{\phi}(P)+\tilde{\phi}(P^{*}))(\tilde{\phi}(P)-\tilde{\phi}(P^{*}))+(\bar{\tilde{A}}_{r+1}
    +\bar{\tilde{D}}_{r+1})(\tilde{\phi}(P)-\tilde{\phi}(P^{*}))\\
    &=\frac{2\tilde{A}_{n+1}\bar{\tilde{B}}_ry}{\tilde{B}_n^2}-(\bar{\tilde{A}}_{r+1}+\bar{\tilde{D}}_{r+1})
    \frac{y}{\tilde{B}_n}\\
  \end{split}
  \end{equation}
  and hence
  \begin{equation}\label{4.58}
    \tilde{\bar{B}}_{n,t_r}=-2\tilde{A}_{n+1}\bar{\tilde{B}}_r+(\bar{\tilde{A}}_{r+1}+\bar{\tilde{B}}_{r+1})\tilde{B}_n.
  \end{equation}
  Similarly, starting from
  (\ref{4.35})
  \begin{equation}\label{4.59}
    \tilde{\phi}_{t_r}(P)+\tilde{\phi}_{t_r}(P^{*})=2z\tilde{C}_r+(\bar{\tilde{A}}+\bar{\tilde{D}}_{r+1})
    \frac{2\tilde{A}_{n+1}}{\tilde{B}_{n}}-\bar{\tilde{B}}_r(\frac{4\tilde{A}_{n+1}^2}{\tilde{B}_n^2}
    +\frac{2z\tilde{C}_n}{\tilde{B}_n})
  \end{equation}
  yields (\ref{4.53}).
  Differentiating the equation (\ref{4.23}) then yields
  \begin{equation}\label{4.60}
    z\tilde{r}_{t_r}\tilde{B}_{n}^--z\tilde{r}\tilde{B}_{n,t_r}^--\tilde{D}^-_{n+1,t_r}-z\tilde{q}_{t_r}\tilde{C}_n-
    z\tilde{q}\tilde{C}_{n,t_r}+\tilde{D}_{n+1,t_r}=0.
  \end{equation}
  Using (\ref{2021}) (\ref{4.17}) (\ref{4.18}) (\ref{4.19}) (\ref{4.21}) (\ref{4.22}) (\ref{4.23}) (\ref{4.52}) and (\ref{4.53}) we have
  \begin{equation}\label{4.61}
  \begin{split}
    &z\tilde{r}_{t_r}\tilde{B}_{n}^--z\tilde{r}\tilde{B}_{n,t_r}^--\tilde{D}^-_{n+1,t_r}-z\tilde{q}_{t_r}\tilde{C}_n-
    z\tilde{q}\tilde{C}_{n,t_r}+\tilde{D}_{n+1,t_r}\\
    &=-z\tilde{B}_{n}^-[\tilde{r}\bar{\tilde{A}}_{r+1}^-+\bar{\tilde{C}}_{r}^--(\tilde{p}-z)\bar{\tilde{C}}_{r}+
    \tilde{r}\bar{\tilde{D}}_{r+1}]-z\tilde{r}[-2\tilde{A}_{n+1}^-\bar{\tilde{B}}_r^-\\
    &+(\bar{\tilde{A}}_{r+1}^-+\bar{\tilde{D}}_{r+1}^-)\tilde{B}_n^-]+z[\tilde{C}_n\bar{\tilde{B}}_r-\tilde{B}_n\bar{\tilde{C}}_r
    -\tilde{C}_n^-\bar{\tilde{B}}_r^-+\tilde{B}_n^-\bar{\tilde{C}}_r^-]\\
    &+z\tilde{C}_n[(\tilde{p}-z)\bar{\tilde{B}}_{r}^-
    -\tilde{q}\bar{\tilde{D}}_{r+1}^--\tilde{q}\bar{\tilde{A}}_{r+1}-\bar{\tilde{B}}_{r}]-z\tilde{q}\tilde{C}_{n,t_r}\\
    &=-z\tilde{B}_{n}^-[-(\tilde{p}-z)\bar{\tilde{C}}_r+\tilde{r}\bar{\tilde{D}}_{r+1}]+z\tilde{r}[-2\tilde{A}_{n+1}^-
    \bar{\tilde{B}}_r^-+\bar{\tilde{D}}_{r+1}^-\tilde{B}_n^-]\\
    &+z[-\tilde{B}_n\bar{\tilde{C}}_r
    -\tilde{C}_n^-\bar{\tilde{B}}_r^-]+z\tilde{C}_n[(\tilde{p}-z)\bar{\tilde{B}}_{r}^-
    -\tilde{q}\bar{\tilde{D}}_{r+1}^--\tilde{q}\bar{\tilde{A}}_{r+1}]-z\tilde{q}\tilde{C}_{n,t_r}\\
    &=z[(\tilde{p}-z)\tilde{B}_n^--\tilde{B}_n]\bar{\tilde{C}}_r+z[-2\tilde{r}\tilde{A}_{n+1}^--\tilde{C}_{n}^-+
    (\tilde{p}-z)\tilde{C}_n]\bar{\tilde{B}}_r^-\\
    &+z[\tilde{r}\tilde{B}_n^-(\bar{\tilde{D}}_{r+1}^--\bar{\tilde{D}}_{r+1})-\tilde{q}\tilde{C}_n
    \bar{\tilde{D}}_{r+1}^-]-z\tilde{q}\tilde{C}_n\bar{\tilde{A}}_{r+1}-z\tilde{q}\tilde{C}_{n,t_r}\\
    &=z\tilde{q}\bar{\tilde{C}}_r[2\tilde{A}_{n+1}+(-\tilde{A}_{n+1}+\tilde{A}_{n+1}^-)]+z\tilde{r}\bar{\tilde{B}}_r^-
    (\tilde{A}_{n+1}-\tilde{A}_{n+1}^-)\\
    &+[(-\tilde{D}_{n+1}+\tilde{D}_{n+1}^-+z\tilde{q}\tilde{C}_n)(\bar{\tilde{D}}_{r+1}^--\bar{\tilde{D}}_{r+1})
    -\tilde{q}\tilde{C}_n\bar{\tilde{D}}_{r+1}^-]\\
    &-z\tilde{q}\tilde{C}_n\bar{\tilde{A}}_{r+1}-z\tilde{q}\tilde{C}_{n,t_r}\\
    &=2z\tilde{q}\tilde{A}_{n+1}\bar{\tilde{C}}_r+(\tilde{A}_{n+1}-\tilde{A}_{n+1}^-)(z\tilde{r}\bar{\tilde{B}}_r^-
    -z\tilde{q}\bar{\tilde{C}}_r)+[(-\tilde{D}_{r+1}+\\
    &\tilde{D}_{r+1}^-+z\tilde{q}\tilde{C}_n)(\bar{\tilde{D}}_{n+1}^--\bar{\tilde{D}}_{n+1})-\tilde{q}\tilde{C}_n
    \bar{\tilde{D}}_{r+1}^-]-z\tilde{q}\tilde{C}_n\bar{\tilde{A}}_{r+1}-z\tilde{q}\tilde{C}_{n,t_r}\\
    &=2z\tilde{q}\tilde{A}_{n+1}\bar{\tilde{C}}_r+(\tilde{A}_{n+1}-\tilde{A}_{n+1}^-)(\bar{\tilde{D}}_{r+1}^-
    -\bar{\tilde{D}}_{r+1})+[(-\tilde{D}_{n+1}+\\
    &\tilde{D}_{n+1}^-+z\tilde{q}\tilde{C}_n)(\bar{\tilde{D}}_{r+1}^--\bar{\tilde{D}}_{r+1})-\tilde{q}\tilde{C}_n
    \bar{\tilde{D}}_{r+1}^-]-z\tilde{q}\tilde{C}_n\bar{\tilde{A}}_{r+1}-\tilde{q}\tilde{C}_{n,t_r}\\
    &=2z\tilde{q}\tilde{A}_{n+1}\bar{\tilde{C}}_r-z\tilde{q}\tilde{C}_n\bar{\tilde{D}}_{r+1}-z\tilde{q}\tilde{C}_n
    \bar{\tilde{A}}_{r+1}-z\tilde{q}\tilde{C}_{n,t_r}=0,\\
  \end{split}
  \end{equation}
  which is equivalent to (\ref{4.54}).
  Finally one can directly differentiate (\ref{4.24})
  \begin{equation}\label{4.60}
    R_{t_r}=-2\tilde{A}_{n+1}\tilde{A}_{n+1,t_{r}}-z\tilde{B}_{n,t_r}\tilde{C}_{n}-z\tilde{B}_{n}\tilde{C}_{n,t_r}
  \end{equation}
  and insert (\ref{4.52})-(\ref{4.54}) into (\ref{4.60}) to yield $R_{t_r}=0.$
  \end{proof}
  \end{proof}
  Next we derive Dubrovin-type equations, that is, first-order ordinary
  differential equations, which govern the dynamics of $\mu_j$ and $\nu_j$
  with respect to variations of $t_r.$ We recall that the affine part of $\mathcal{K}_n$ is
  nonsingular if
  \begin{equation}\label{4.61}
    \{E_m\}_{m=0,\ldots,2n+1}\subset\mathbb{C},\quad E_m\neq E_{m^{'}}\quad\text{for}\quad m\neq m^{'},m,m^{'}=0,\ldots,2n+1.
  \end{equation}
  \begin{lemma}
  Assume Hypothesis 2 and Hypothesis 3 and
  suppose (\ref{4.2}) (\ref{4.4}) hold on $\mathbb{Z}\times\Omega_\mu$ with $\Omega_\mu\subseteq\mathbb{R}$
  an open interval. In addition, assume that the zeros $\mu_j$ of $\tilde{B}_n(\cdot)$ remain distinct and nonzero on $\mathbb{Z}\times\Omega_{\mu}$. Then $\{\hat{\mu}_j\}_{j=1,\dotsi,n}$ defined in (\ref{4.25}) satisfy the following first-order system of differential system on $\mathbb{Z}\times\Omega_{\mu}$,
  \begin{equation}\label{4.62}
    \mu_{j,t_r}=\frac{(\tilde{q^+})^{-1}y(\hat{\mu_j})\bar{\tilde{B}}_r(\mu_j)}{\prod_{k=1,k\neq j}^{n}(\mu_k-\mu_j)},\quad j=1,\dotsi,n,
  \end{equation}
  with $$\hat{\mu}_j(n,\cdot)\in\textit{C}^{\infty}(\Omega_{\mu},\mathcal{K}_n),\quad j=1,\ldots,n, n\in\mathbb{Z}.$$
 For the zeros $\nu_j$ of $\tilde{C}_n,$ identical statements hold with $\mu_j$ and $\Omega_\mu$ replaced by
 $\nu_j$ and $\Omega_\nu,$ etc. In particular, $\{\hat{\nu}_j\}_{j=1,\dotsi,n},$ defined in (\ref{4.25}), satisfies the following first-order system on $\mathbb{Z}\times\Omega_\nu,$
 \begin{equation}\label{4.63}
   \nu_{j,t_r}=\frac{2\tilde{r}^{-1}y(\hat{\nu}_j)\bar{\tilde{C}}_r(\nu_j)}{\nu_j\prod_{k=1,k\neq j}^{n}(\mu_k-\mu_j)},\quad j=1,\dotsi,n,
 \end{equation}
 with $$\hat{\nu}_j(n,\cdot)\in\textit{C}^{\infty}(\Omega_{\nu},\mathcal{K}_n),\quad j=1,\ldots,n, n\in\mathbb{Z}.$$

 \end{lemma}
 \begin{proof}
 It suffices to consider (\ref{4.62}) for $\mu_{j,t_r}.$
 Using the product representation for $\tilde{B}_{n}$ in (\ref{4.24a})
 $$\tilde{B}_{n}=(-\tilde{q}^{+})\prod_{j=1}^{n}(\lambda-\mu_j)$$
and employing the (\ref{4.25}) and (\ref{4.52}), one computes
\begin{align*}
\begin{split}
&\tilde{B}_{n,t_r}(\mu_j)=-2\tilde{A}_{n+1}(\mu_j)\bar{\tilde{B}}_{r}(\mu_j)\\
&=y(\hat{\mu}_j)\bar{\tilde{B}}_{r}(\mu_j)
=\tilde{q}^+\mu_{j,t_r}\prod_{k\neq j,j=1}^n\left(\mu_j-\mu_k\right),\quad j=1,\dotsi,n,\\
\end{split}
\end{align*}
proving (\ref{4.62}). The case of (\ref{4.63}) for $\nu_j$ is analogous using the product presentation
for $\tilde{C}_n$ in (\ref{4.24a}) and employing (\ref{4.25}) and (\ref{4.54}).

 \end{proof}
Since the stationary trace formulas for $\tilde{a}_\ell, \tilde{b}_\ell$ in terms of symmetric functions of $\mu_j$ and $\nu_j$ in Lemma \ref{lemma2} extend line by line to the corresponding time-dependent setting, we next record their $t_r$-dependent analogs without proof. For simplicity we confine ourselves to the simplest ones only.
\begin{lemma}
Assume Hypothesis 2 and Hypothesis 3 hold and
suppose $p(n,t_r),q(n,t_r)$ satisfy (\ref{4.2}) (\ref{4.4}). Then, we have the following trace formula
\begin{equation}\label{4.64}
  \tilde{p}^{+}-\tilde{q}^{+}\tilde{r}-\tilde{q}^{++}\tilde{r}^+-\tilde{q}^{++}/\tilde{q}^{+}+\delta_1=-\sum_{j=1}^{n}\mu_j,
\end{equation}
\begin{equation}\label{4.65}
  \tilde{p}-\tilde{q}^{+}\tilde{r}-\tilde{r}^{-}/\tilde{r}-\tilde{q}\tilde{r}^{-}+\delta_1=-\sum_{j=1}^{n}\nu_j.
\end{equation}

\end{lemma}
\begin{lemma}\label{lemma8}
Suppose $p,q$ satisfy the (\ref{4.1}) (\ref{4.2}) and the $n$-th stationary RLV system (\ref{2.34}). Moreover, let $(n,t_r)\in\mathbb{Z}\times\mathbb{R}$ and $\mathcal{D}_{\underline{\hat{\mu}}}, \underline{\hat{\mu}}=\left(\hat{\mu}_1,\dotsi,\hat{\mu}_{n}\right)\in \text{Sym}^n(\mathcal{K}_n)$, $\mathcal{D}_{\underline{\hat{\nu}}},\underline{\hat{\nu}}=\left(\hat{\mu}_1,\dotsi,\hat{\nu}_{n}\right)\in \text{Sym}^n(\mathcal{K}_n)$ be the pole and zero divisors of degree $n$, respectively, associated with $p, q$ and $\tilde{\phi}$ defined according to (\ref{4.26}), that is,
\begin{equation*}
\hat{\mu}_j(n,t_{r})=\left(\mu_j(n,t_r),-2\tilde{A}_{n+1}(\mu_j(n,t_r),n,t_r)\right)\in\mathcal{K}_n,\quad j=1,\dotsi,n,
\end{equation*}
\begin{equation*}
\hat{\nu}_j(n,t_{r})=\left(\nu_j(n,t_r),2\tilde{A}_{n+1}(\nu_j(n,t_r),n,t_r)\right)\in\mathcal{K}_n,\quad j=1,\dotsi,n.
\end{equation*}
Then $\mathcal{D}_{\underline{\hat{\mu}}(n,t_r)}$ and $\mathcal{D}_{\underline{\hat{\nu}}(n,t_r)}$  are non-special for all $(n,t_{r})\in\mathbb{Z}\times\mathbb{R}$.
\end{lemma}
\begin{proof}
We are only to prove the conclusion for $\mathcal{D}_{\underline{\hat{\mu}}(n,t_r)}$.
$\mathcal{D}_{\underline{\hat{\mu}}(n)}$ is non-special if and
only if $\{\hat{\mu}_1(n),\dotsi,\hat{\mu}_n(n)\}$ contains one
pair of $\{\hat{\mu}_j,\hat{\mu}^*_j\}$. Hence, $\mathcal{D}_{\underline{\hat{\mu}}(n)}$
is non-special as long as the projection $\mu_j$ of $\hat{\mu}_j$ are mutually distinct,
$\mu_j(n)\neq\mu_k(n)$ for $j\neq k$. If two or more projection coincide for some $n_0\in\mathbb{Z}$,
for instance,$$\mu_{j_1}(n_0)=\dotsi=\mu_{j_k}(n_0)=\mu_0,\quad k>1.$$ There are two cases in the
following associated with $\mu_0$. If $\mu_0\in\mathbb{C}\backslash\{E_0,\dotsi,E_{2n+1}\}$, then $\tilde{A}_{n+1}(\mu_0,n_0,t_r)\neq 0$. It is obvious that
$\hat{\mu}_{j_1}(n_0,t_r),\dotsi,\hat{\mu}_{j_k}(n_0,t_r)$ all meet in the same sheet and hence no special divisor
can arise in this manner. If $\mu_0$ equals to some $E_{m_0}$ and $k>1$, one concludes
$$\tilde{B}_{n}(z,n_0,t_r)\equfill{z\rightarrow E_{m_0}}{}O\left((z-E_{m_0})^2\right)$$ and $$\tilde{A}_{n+1}(E_{m_0},n_0,t_r)=0.$$
But one observes $R_{2p+2}(z,n_0,t_r)=-\tilde{A}_{n+1}^2-z\tilde{B}_n\tilde{C}_n=O\left((z-E_{m_0})^2\right).$
This conclusion contradict with the hypothesis that the curve is nonsingular. We have $k=1 $. Therefore no special divisor can arise in this manner. Then we have completed the proof.
\end{proof}

Next, Now we turn to asymptotic properties of $\tilde{\phi}, \Psi_1, \Psi_2$ defined in (\ref{4.26}) (\ref{4.28}) and (\ref{4.29}) in a neighborhood of $P_{0,\pm}$ and $P_{\infty\pm}.$
\begin{lemma}
Assume Hypothesis 2 and Hypothesis 3 hold and
suppose $p(n,t_r),q(n,t_r)$ satisfy (\ref{4.2}) (\ref{4.4}).
Moreover, let $P=(z,y)\in\mathcal{K}_{n}\backslash\{P_{\infty\pm},P_{0}\}$, $(n,n_{0},t_r)\in\mathbb{Z}\times\mathbb{Z}\times\mathbb{R}$. Then, the meromorphic function $\tilde{\phi}$ on
$\mathcal{K}_n$ has the following asymptotic behavior
\begin{equation}\label{4.66}
  \tilde{\phi}(P)=
  \begin{cases}
      (\tilde{q}^{+})^{-1}\zeta^{-1}+\left(((\tilde{q}^{+}/\tilde{q}-1)\tilde{p})/\tilde{q}\right)^{+}+O(\zeta)
                                                                             &\text{as}\quad P\rightarrow P_{\infty+}, \cr
      -\tilde{r}+(\tilde{p}\tilde{r}^--\tilde{p}\tilde{r})\zeta+O(\zeta^2)
                                                                              &\text{as}\quad P\rightarrow P_{\infty-}, \cr
  \end{cases}
\end{equation}
where we use the local coordinate $z=\zeta^{-1}$ near the points $P_{\infty\pm}$.
\begin{equation}\label{4.67}
    \tilde{\phi}(P)=
  \begin{cases}
     c_{n}/(\prod_{m=0}^{2n+1}E_{m})\zeta+O(\zeta^2)
                                                                             &\text{as}\quad P\rightarrow P_{0,+}, \cr
      -\left(\prod_{m=0}^{2n+1}E_m\right)/b_n+O(\zeta)
                                                                              &\text{as}\quad P\rightarrow P_{0,-}, \cr
  \end{cases}
\end{equation}
where we use the local coordinate $z=\zeta$ near the points $P_{0,\pm}.$

\end{lemma}
\begin{proof}
By the definition of $\tilde{\phi}$ in (\ref{4.26}) the time parameter $t_r$ can be viewed as an additional but
fixed parameter, the asymptotic behavior of $\tilde{\phi}$ remains the same as in Lemma \ref{lemma3}.
\end{proof}

\begin{theorem}
Assume Hypothesis 2 and Hypothesis 3 hold and
suppose $p(n,t_r),q(n,t_r)$ satisfy (\ref{4.2}) (\ref{4.4}). Moreover, let $P\in\mathcal{K}_{n}$\textbackslash$\{P_{\infty\pm},P_{0,\pm}\}$
and $(n,n_{0})\in\mathbb{Z}^{2}.$ Then $\mathcal{D}_{\underline{\hat{\mu}}(n,t_r)}$ is non-special. Moreover,\\
\begin{equation} \label{4.69}
\phi(P,n,t_r)=C(n,t_r)\frac{\theta(\underline{z}(P,\underline{\hat{\nu}}(n,t_r)))}{\theta(\underline{z}(P,\underline{\hat{\mu}}(n,t_r)))}\exp\left(\int_{Q_{0}}^{P}\omega_{P_{0,+}P_{\infty+}}^{(3)}\right),
\end{equation}
and $p, q$ are the form of
\begin{equation}\label{4.70}
  p^{+}(n,t_r)=\frac{1}{2}\left(-\Delta_3-\Delta_3^{+}-\delta_1+1-\Delta_1\pm\left((\Delta_3+\Delta_3^++\delta_1-1+\Delta_1)^{2}+4\Delta_2\right)^{\frac{1}{2}}\right)
\end{equation}
\begin{equation}\label{4.71}
  \begin{split}
  &q^{+}(n,t_r)=\Delta_2/p^{+}+1\\
  &=2\Delta_2\left(-\Delta_3-\Delta_3^{+}-\delta_1+1-\Delta_1\pm\left((\Delta_3+\Delta_3^++\delta_1-1+\Delta_1)^{2}+4\Delta_2\right)^{-\frac{1}{2}}\right)^{-1}\\
  &+1.\\
  \end{split}
\end{equation}
Here
\begin{equation}\label{4.72}
  \Delta_1=\sum_{j=1}^n\lambda_j^{'}-\sum_{j=1}^n c_j(k)\partial_{\omega_j}\ln\left(\frac{\theta(\underline{z}(P_{\infty+},\underline{\hat{\mu}}(n,t_r))
  +\underline{\omega})}{\theta(\underline{z}(P_{\infty-},\underline{\hat{\mu}}(n,t_r))+\underline{\omega})}\right)|_{\underline{\omega}=0},
\end{equation}
\begin{equation}\label{4.73}
  \Delta_2=\kappa_{\infty+}-\sum_{j=1}^{n}c_{j}(n)\partial_{\omega_{j}}   \ln\left(\frac{\theta(\underline{z}(P_{\infty+},\underline{\hat{\nu}}(n,t_r))+\underline{\omega})}
   {\theta(\underline{z}(P_{\infty+},\underline{\hat{\mu}}(n,t_r))+\underline{\omega})}\right)
   |_{\underline{\omega}=0},
\end{equation}
\begin{equation}\label{4.74}
  \Delta_3=\frac{\theta(\underline{z}(P_{\infty-},\underline{\hat{\nu}}(n,t_r)))}
  {\theta(\underline{z}(P_{\infty-},\underline{\hat{\mu}}(n,t_r)))}\frac{\theta(\underline{z}(P_{\infty+},\underline{\hat{\mu}}(n,t_r)))}
  {\theta(\underline{z}(P_{\infty+},\underline{\hat{\nu}}(n,t_r)))}\frac{\tilde{a}_2}{\tilde{a}_1}
\end{equation}
and $\tilde{a}_1, \tilde{a}_2$, $ \{\lambda_j^{'}\}_{j=1,\dotsi,n}\in\mathbb{C}$ in (\ref{3.32}).
The Abel map evolving with respect to $t_r$ are
\begin{equation}\label{4.74a}
  \frac{\partial}{\partial_{t_r}}\underline{\alpha}_{Q_{0},\ell}(\mathcal{D}_{\underline{\hat{\mu}}(n,t_r)})
  =\frac{\partial}{\partial_{t_r}}\underline{\alpha}_{Q_{0},\ell}(\mathcal{D}_{\underline{\hat{\nu}}(n,t_r)})
=-\sum_{k=1}^{p}\sum_{s=0}^{r}c_\ell(k)\bar{\tilde{\delta}}_{r-s}\hat{c}_{k+s-p}(\underline{E}).
\end{equation}

\end{theorem}
\begin{proof}
The proof of (\ref{4.69})-(\ref{4.74})  is analogous with theorem \ref{TH4}. Here $t_r$ can be regarded as a parameter. Next we prove (\ref{4.74a}).
\begin{equation}
\begin{split}
&\frac{\partial}{\partial_{t_r}}\underline{\alpha}_{Q_{0},\ell}(\mathcal{D}_{\underline{\hat{\mu}}(n,t_r)})
=\frac{\partial}{\partial_{t_r}}\sum_{j=1}^n\int_{Q_0}^{\hat{\mu}_j(n,t_r)}\omega_\ell\\
&=\sum_{j=1}^n \omega_\ell(\hat{\mu}_j)\mu_{j,t_r}\\
&=\sum_{j=1}^n\left(\sum_{k=1}^n c_\ell(k)\frac{\mu_j^{k-1}}{y(\hat{\mu}_j(n,t_r))}\right)\left((\tilde{q}^{+})^{-1}\bar{\tilde{B}}_{r}(\mu_{j}(n,t_r))
y(\hat{\mu}_j(n,t_r))\prod_{k=1,k\neq j}^{n}\left(\mu_{j}-\mu_{k}\right)^{-1}\right)\\
&=\sum_{j=1}^n\left(\sum_{k=1}^n c_\ell(k)\frac{\mu_j^{k-1}}{\prod_{k=1,k\neq j}^{n}\left(\mu_{j}-\mu_{k}\right)}\right)\left((\tilde{q})^{-1}\bar{\tilde{B}}_{r}(\mu_{j}(n,t_r))\right)\\
&=-\sum_{k=1}^{n}c_\ell(k)\sum_{j=1}^{n}\frac{\mu_j^{k-1}}{\prod_{k=1,k\neq j}^{n}\left(\mu_{j}-\mu_{k}\right)}\left(\sum_{s=0}^{r}\bar{\tilde{\delta}}_{r+1-s}\left(\sum_{t=\text{max}\{0,s-p\}}^s\hat{c}_t(\underline{E})\Psi_{s-t}^{(j)}(\underline{\mu})\right)\right)\\
&=-\sum_{k=1}^{p}\sum_{s=0}^{r}c_\ell(k)\bar{\tilde{\delta}}_{r-s}\hat{c}_{k+s-p}(\underline{E}).\\
\end{split}
\end{equation}
\end{proof}

\section{Appendix: The Lagrange Interpolation Representation of $-(\tilde{q}^{+})^{-1}\tilde{B}_{r+1}(\mu_j(n,t_r))$}

We search for the interpolation representation of $\tilde{B}_r(\mu_j(n,t_r))$ as in the KdV, AKNS, Toda cases.
Introducing the notation in \cite{A1,A2} ,
\begin{align*}
\begin{split}
&\Psi_k(\underline{\mu})=(-1)^k\sum_{\underline{\ell}\in\mathcal{S}_k}\mu_{\ell_1}\dotsi,\mu_{\ell_k},\\ &\mathcal{S}_k=\{\underline{l}=(\ell_1,\dotsi,\ell_k)\in\mathbb{N}^k|\ell_1<\dotsi<\ell_k\leq n\},\quad k=1,\dotsi n,\\
&\Phi_k^{(j)}(\underline{\mu})=(-1)^k\sum_{\underline{\ell}\in\mathcal{\tau}_k^{(j)}}\mu_{\ell_1}\dotsi,\mu_{\ell_k},\\ &\mathcal{\tau}_k^{(j)}=\{\underline{l}=(\ell_1,\dotsi,\ell_k)\in\mathbb{N}^k|\ell_1<\dotsi<\ell_k\leq n,\quad \ell_m\neq j,\quad m=1,\dotsi,k\}, \\
&k=1,\dotsi n-1,\quad j=1,\dotsi,n\\
\end{split}
\end{align*}
and the formula
\begin{equation}\label{6.3}
\sum_{\ell=0}^{k}\Psi_{k-\ell}(\underline{\mu})\mu_j^\ell=\Phi_{k}^{(j)}(\underline{\mu}),\quad k=0,\dotsi,n,\quad j=1,\dotsi,n.
\end{equation}
Let $B_n=(\tilde{q}^{+})^{-1}\tilde{B}_n, b_s=(\tilde{q}^{+})^{-1}\tilde{b}_s(s=0,1,\dotsi,n),$
then one finds
$$B_{n}(z)=\sum_{\ell=0}^{n}b_{n-s}z^s=\prod_{j=1}^n\left(z-\mu_j\right)=\sum_{l=0}^{n}\Psi_{p-l}(\underline{\mu})
z^\ell$$
and  $$b_\ell=\Psi_{\ell}\left(\underline{\mu}\right),\quad \ell=0,\dotsi,n.$$
In the case $r<n$,
\begin{align}\label{6.4}
\begin{split}&\bar{B}_{r}=\sum_{s=0}^{r}\bar{b}_{r-s}z^s=\sum_{s=0}^{r}\left(\sum_{k=0}^{\text{min}\{r-s,n\}}\hat{c}_{r-s-k}
(\underline{E})b_k\right)z^s\\
&=\sum_{s=0}^{r}\left(\sum_{k=0}^{r-s}\hat{c}_{r-s-k}(\underline{E})b_k\right)z^s\\
&=\sum_{s=0}^{r}\left(\sum_{k=0}^{r-s}\hat{c}_{r-s-k}(\underline{E})\Psi_{k-1}(\underline{\mu})\right)z^s\\
 &=\sum_{s=0}^{r}\hat{c}_{s}(\underline{E})\sum_{t=0}^{r-s}\Psi_{r-s-t}(\underline{\mu})z^t.\\
\end{split}
\end{align}
Using (\ref{6.3}), we have
\begin{equation}\label{6.5}
\begin{split}
&\bar{B}_{r}(\mu_j)=\sum_{s=0}^{r}\hat{c}_{s}(\underline{E})\sum_{t=0}^{r-s}\Psi_{r-s-t}(\underline{\mu})\mu_j^t\\
&=\sum_{s=0}^{r}\hat{c}_{s}(\underline{E})\Psi_{r-s}^{(j)}(\underline{\mu}).\\
\end{split}
\end{equation}
In the case $r>n$,
\begin{align}
\begin{split}
&\bar{B}_{r}(z)=\sum_{s=0}^{r}\bar{b}_{r-s}z^s=\sum_{s=0}^{r}\left(\sum_{k=0}^{\text{min}\{r-s,n\}}\hat{c}_{r-s-k}
(\underline{E})b_k\right)z^s\\
&=\sum_{s=0}^{r-n}\sum_{k=0}^{n}\hat{c}_{r-s-k}(\underline{E})\Psi_{k}(\underline{\mu})z^s+\sum_{s=r-p+1}^{r}
\sum_{k=0}^{r-s}\hat{c}_{r+1-s-k}(\underline{E})\Psi_{k}(\underline{\mu})z^s\\
&=\sum_{s=0}^{r-n}\sum_{k=0}^{n}\hat{c}_{r-s-k}(\underline{E})\Psi_{k}(\underline{\mu})z^s+
\sum_{s=r-n+1}^{r}\sum_{k=0}^{n}\hat{c}_{r-s-k}(\underline{E})\Psi_{k}(\underline{\mu})z^s\\
&=\sum_{k=0}^{n}\sum_{s=0}^{r}\hat{c}_{r-s-k}(\underline{E})\Psi_{k}(\underline{\mu})z^s\\
&=\sum_{s=0}^{r}\sum_{k=0}^{n}\hat{c}_{r-s-k}(\underline{E})\Psi_{k}(\underline{\mu})z^s\\
&=\sum_{s=0}^{r}\sum_{k=0}^{n}\hat{c}_{s}(\underline{E})\Psi_{k}(\underline{\mu})z^{r-s-k}\\
&=\sum_{s=0}^{r-n}\hat{c}_{s}(\underline{E})\left(\sum_{k=0}^{n}\Psi_{k}(\underline{\mu})z^{p-k}\right)z^{r-n-s}
+\sum_{s=r-p+1}^{r}\hat{c}_{s}(\underline{E})\left(\sum_{k=0}^{n}\Psi_{k}(\underline{\mu})z^{r-s-k}\right)\\
&=\sum_{s=0}^{r-n}\hat{c}_{s}(\underline{E})\left(B_{n}(\lambda)\right)z^{r-n-s}+\sum_{s=r-n+1}^{r}\hat{c}_{s}
(\underline{E})\left(\sum_{k=0}^{r-s}\Psi_{k}(\underline{\mu})z^{r-s-k}\right).\\
\end{split}
\end{align}
Then we have
\begin{equation}\label{6.7}
\bar{B}_{r}(\mu_j)=\sum_{s=r-n+1}^{r}\hat{c}_{s}(\underline{E})\left(\sum_{k=0}^{r-s}\Psi_{k}(\underline{\mu})
\mu_j^{r-s-k}\right)=\sum_{s=r-n+1}^{r}\hat{c}_s(\underline{E})\Psi_{r-s}^{(j)}(\underline{\mu}).
\end{equation}
Combining (\ref{6.5}) with (\ref{6.7}), one finds
\begin{equation}
\bar{B}_{r}(z)=\sum_{s=\text{max}\{0,r-n\}}^{r}\hat{c}_{s}(\underline{E})\Psi_{r-s}^{(j)}(\underline{\mu}).
\end{equation}
Hence
\begin{equation}
\begin{split}
&\bar{\tilde{B}}_{r}(\mu_j)=\sum_{s=0}^{r}\bar{\tilde{\delta}}_{r-s}\bar{B}_s(\mu_j)=\sum_{s=0}^{r}
\tilde{\delta}_{r-s}\left(\sum_{t=\text{max}\{0,s-n\}}^{s}\hat{c}_{t}(\underline{E})\Psi_{s-t}^{(j)}
(\underline{\mu})\right).\\
\end{split}
\end{equation}

\renewcommand{\baselinestretch}{1.2}

\end{document}